\documentclass[runningheads]{llncs}
\usepackage{fullpage}
\usepackage{graphicx}
\usepackage[usenames,dvipsnames]{xcolor}
\usepackage{makeidx}
\usepackage{algorithm}
\usepackage{algorithmic}
\usepackage{graphicx,tipa,subfigure}
\usepackage{arcs,lmodern,fix-cm}
\usepackage{times}
\usepackage{amsmath,amsfonts}
\usepackage{slashbox,multirow}
\usepackage{rotating}
\usepackage{verbatim}
\usepackage{caption}
\usepackage{comment}
\usepackage{caption}

\newtheorem{observation}[theorem]{Observation}

\newcommand{\reals}{\mathbb{R}}
\newcommand{\todo}[1]{{\color{blue}\textsl{\small[#1]}\marginpar{\small\textsc{\textbf{To do!}}}}}
\newcommand{\ignore}[1]{}

\newcommand{\tann}[1]{\tan\left(#1\right)}

\newcommand{\sinn}[1]{\sin \left({#1}\right)}

\newcommand{\coss}[1]{\cos \left({#1}\right)}

\def\qed{\hfill\rule{2mm}{2mm}}

\def\uu{\overline{\mathcal U}(a,b,c)}
\def\ul{\underline{\mathcal U}(a,b,c)}
\def\ll{\underline{\mathcal L}(a,b,c)}
\def\lu{\overline{\mathcal L}(a,b,c)}

\def\os{\textsc{OppositeSearch}}

\excludecomment{showproof}

\begin{document}
\title{Triangle Evacuation of 2 Agents in the Wireless Model
\thanks{
This is the full version of the paper with the same title \cite{GJ22-triangle-algosensros} which will appear in the proceedings of the 
18th International Symposium on Algorithms and Experiments for Wireless Sensor Networks (ALGOSENSORS 2022), 8-9 September 2022 in Potsdam, Germany.
}
}
%
%
\author{
Konstantinos Georgiou
\thanks{Research supported in part by NSERC.}
\and
Woojin Jang
\ignore{
First Author\inst{1}
\and
Second Author\inst{2,3}
\and
Third Author\inst{3}
} 
}
\authorrunning{K. Georgiou and W. Jang}
%
\institute{
Department of Mathematics, Toronto Metropolitan University,
Toronto, ON, M5B 2K3, Canada
\email{\{konstantinos,woojin.jang\}@ryerson.ca}
}
\maketitle              
\begin{abstract}
The input to the \emph{Triangle Evacuation} problem is a triangle $ABC$. 
Given a starting point $S$ on the perimeter of the triangle, a feasible solution to the problem consists of two unit-speed trajectories of mobile agents that eventually visit every point on the perimeter of $ABC$. The cost of a feasible solution (evacuation cost) is defined as the supremum over all points $T$ of the time it takes that $T$ is visited for the first time by an agent plus the distance of $T$ to the other agent at that time. 

Similar evacuation type problems are well studied in the literature covering the unit circle, the $\ell_p$ unit circle for $p\geq 1$, the square, and the equilateral triangle. 
We extend this line of research to arbitrary non-obtuse triangles. Motivated by the lack of symmetry of our search domain, we introduce 4 different algorithmic problems arising by letting the starting edge and/or the starting point $S$ on that edge to be chosen either by the algorithm or the adversary. To that end, we provide a tight analysis for the algorithm that has been proved to be optimal for the previously studied search domains, as well as we provide lower bounds for each of the problems. Both our upper and lower bounds match and extend naturally the previously known results that were established only for equilateral triangles. 

\keywords{
Search
\and
Evacuation
\and
Triangle
\and
Mobile Agents
}
\end{abstract}

\section{Introduction}

Search Theory is concerned with the general problem of retrieving information 
in some search domain.
Seemingly simple but mathematically rich problems pertaining to basic geometric domains where studied as early as the 1960's, see for example~\cite{beck1964linear} and~\cite{bellman1963optimal}, then popularized in the Theoretical Computer Science community in the late 1980's by the seminal works of Baeza-Yates et al.~\cite{baeza1988searching}, and then recently re-examined under the lens of mobile agents, e.g. in~\cite{CzyzowiczGGKMP14}. 

In search problems closely related to our work, a hidden item (exit) is placed on a geometric domain and can only be detected by unit speed searchers (mobile agents or robots) when any of them walks over it. When the exit is first visited by an agent, its location is communicated instantaneously to the remaining agents (wireless model) so that they all gather to the exit along the shortest path in the underlying search domain. The time it takes for the last agent to reach the exit, over all exit placements, is known as the \emph{evacuation time} of the (evacuation) algorithm. 

Optimal evacuation algorithms for 2 searchers are known for a series of geometric domains, e.g. the circle~\cite{CzyzowiczGGKMP14}, $\ell_p$ circles~\cite{georgiou2021evacuating},
and the square and the equilateral triangle~\cite{CzyzowiczKKNOS15}.
For all these problems, a plain-vanilla algorithm is the optimal solution: the two searchers start from a point on the perimeter of the geometric domain and search the perimeter in opposite directions at the same unit speed until the exit is found, at which point the non-finder goes to the exit along the shortest path. 

In this work we consider the 2 searcher problem in the wireless model over non-obtuse triangles. The consideration of general triangles comes with unexpected challenges, since even the worst-case analysis of the plain-vanilla algorithm becomes a technical task due to the lack of symmetry of the search domain. But even more interestingly, the consideration of general triangles allows for the introduction of 4 different algorithmic problems pertaining to the initial placement of the searchers. Indeed, in all previously studied domains, the optimal algorithm had the searchers start from the same point on the perimeter. 

Motivated by this, we consider 4 different algorithmic problems that we believe are interesting in their own right. These problems are obtained by specifying whether the algorithm or the adversary chooses the starting edge and/or the starting point on that edge for the 2 searchers. In our results we provide upper bounds for all 4 problems, by giving a technical and tight analysis of the plain-vanilla algorithm that we call \os. We also provide lower bounds for arbitrary search algorithms for the same problems. 
This is the full version of a published extended abstract with the same tile~\cite{GJ22-triangle-algosensros}.

It is worth noting that the results of~\cite{CzyzowiczKKNOS15} for the equilateral triangle were presented in the form of evaluating the evacuation time of the \os~ algorithm for any starting point on the perimeter for the purpose of finding only the best starting point. It is an immediate corollary to obtain optimal results for all 4 new algorithmic problems, but only for the equilateral triangle. Both our upper and lower bounds for general triangles match and extend the aforementioned results for the equilateral triangle. 
To that end, one of our surprising findings is that for some of the algorithmic problems, the evacuation time is more than half the perimeter of the search domain, even though the induced searchers' trajectories stay exclusively on the perimeter of the domain and the searchers cover the entire triangle perimeter for the worst placements of the exit (which would be the same trajectory, should the exit was found only at the very end of the search). This is despite that two agents can search the entire domain in time equal to half the triangle perimeter.

\section{Related Work}

Search problems have received several decades of treatment that resulted in some interesting books~\cite{alpern2013search,AlpGal03}, in a relatively recent survey~\cite{hohzaki2016search}, and in an exposition of mobile agent-based (and hence more related to our work) results in~\cite{11340}.
A primitive, and maybe the most well-cited search-type problem is the so-called cow-path or linear-search~\cite{baezayates1993searching} in which a mobile agent searches for hidden item on the infinite line. Numerous variations of the problem have been studied, ranging from different domains to different searchers specifications to different objectives. The interested reader may consult any of the aforementioned resources for problem variations. Below we give a representative outline of results most relevant to our new contributions, and pertaining primarily to wireless evacuation from geometric domains.  


The problem of evacuating 2 or more mobile agents from a geometric domain was consider by Czyzowicz et al.~\cite{CzyzowiczGGKMP14} on the circle. The authors considered the distinction between the wireless (our communication model), where searchers exchange information instantaneously, and the face-to-face model, where information cannot be communicated from distance. The results pertaining to the wireless model were optimal, while improved algorithms for the face-to-face model followed in~\cite{CGKNOV20,brandt2017collaboration,disser2019evacuating}.
A little later, the work of~\cite{chuangpishit2020multi} also considered trade-offs to multi-objective evacuations costs. 
Since the first study of~\cite{CzyzowiczGGKMP14} a number of variations emerged. Different searchers' speeds were considered in~\cite{lamprou2016fast}, 
faulty robots were studied in~\cite{CzyzowiczGGKKRW17,GeorgiouKLPP19}, evacuation from multiple exits was introduced in~\cite{CzyzowiczDGKM16},
searching with advice was considered in~\cite{GeorgiouKS17}, and search-and-fetch type evacuation was studied in~\cite{kranakis2019search}.

The work most relevant to our new contributions are the optimal evacuation results in the wireless model of Czyzowicz et al.~\cite{CzyzowiczKKNOS15} pertaining to the equilateral triangle. In particular, Czyzowicz et al. considered the problem of placing two unit speed wireless agents on the perimeter of an equilateral triangle. If the starting point is $x$ away from the middle point of any edge (assumed to have unit length, hence $x\leq 1/2$), they showed that the optimal worst-case evacuation cost equals $3/2+x$. As a corollary, the overall best evacuation cost is 3/2 if the algorithm can choose the starting point, even if an adversary chooses first an edge for where the starting point will be placed. Similarly, the evacuation cost is 2 when an adversary chooses the searchers' starting point even if the algorithm can choose the starting edge. These are the results that we extend to general triangles by obtaining formulas that are functions of the triangle edges. Notably, the search domains of~\cite{CzyzowiczKKNOS15} were later re-examined in the face-to-face model~\cite{ChuangpishitMNO20} and with limited searchers' communication range~\cite{BagheriNO19}. 

One of the main features of our results is that we quantify the effectiveness of choosing the starting point of the searchers. When the search domain is the circle~\cite{CzyzowiczKKNOS15}, the starting point is irrelevant as long as both agents start from the same point (and turns out it is optimal to have searchers initially co-located on the perimeter). Nevertheless, strategic starting points based on the searchers' relative distance have been considered in~\cite{czyzowicz2020priority123,czyzowicz2020priority4}. Apart from~\cite{CzyzowiczKKNOS15}, the only other result we are aware of where the (same) starting point of the searchers has to be chosen strategically is that of~\cite{georgiou2021evacuating} that considered evacuation from $\ell_p$ unit discs, with $p\geq 1$. 
Nevertheless, a key difference in our own result is that the geometry of the underlying search domain, that is a triangle, allows us to quantify how good an evacuation solution can be when choosing the searchers' and the exit's placement. In particular the order of decisions for the problems studied here, once one fixes an arbitrary triangle, are obtained by fixing (i) the starting edge of the searchers, then (ii) the starting point on that edge, then (iii) the searchers' evacuation algorithm, and then (iv) the placement of the exit, in that order. Item (iii) is always an algorithmic choice while (iv) is always an adversarial choice.

\section{Problem Definition, Motivation \& Results}
\label{sec: problem definition and contributions}

We begin with the problem definition and some motivation. 
 Two unit speed mobile agents start from the same point on the perimeter of some non-obtuse triangle $ABC$. Somewhere on the perimeter of $ABC$ there is a hidden object, also referred to as the \textit{exit}, that an agent can see only if it is collocated with the exit. When an agent identifies the location of the exit, the information reaches the other agent instantaneously, or as it is described in the literature, the mobile agents operate under the so-called wireless model. 
A feasible solution to the problem consists of agents' trajectories that ensure that regardless of the placement of the exit, both agents eventually reach the exit. For this reason such solutions/algorithms are also known as evacuation algorithms. 
The time it takes the last agent to reach the exit is known as the \emph{evacuation time} of the algorithm for the specific exit placement. 
As it is common for worst-case analysis, in this work we are concerned with the calculation of the \emph{worst-case evacuation cost} of an algorithm (or simply evacuation cost), that is, the supremum of the evacuation time  over all possible placements of the exit.
Note that the worst-case evacuation cost of an algorithm is a function of the agents' starting point, and of course the size of the given triangle $ABC$. 

The 2-agent evacuation problem in the wireless model has been studied on the circle, the square and the equilateral triangle. Especially when the search domain enjoys the symmetry of the circle, the starting point of the agents is irrelevant. In other cases when the search domain exhibits enough symmetries, e.g. for the square and the equilateral triangle, the exact evaluation of the worst-case evacuation cost (as a function of the agents' starting point) for a natural (and optimal algorithm) is an easy exercise, and due to the underlying symmetries not all starting points need to be considered. Interestingly, things are quite different in our problem when the search domain is a general non-obtuse triangle; indeed, not only the worst-case evacuation cost depends (in a strong sense) on the agents' starting point, but also its' exact evaluation (as a function of the starting point, and the triangle's edges) is a non-trivial task even for a plain-vanilla evacuation algorithm. 

Due to the asymmetry of a general triangle, 
we are motivated to identify efficient evacuation algorithms when the agents' starting point is chosen in two steps: \\
\emph{Step 1:} Choose a starting edge (either the largest, or the second largest, or the smallest), and \\
\emph{Step 2:} Choose the starting point on that edge. \\
Each of the two choices can be either \emph{algorithmic or adversarial}, giving rise to 4 different interesting algorithmic questions, summarized in the next table, that also introduces notation for the best possible evacuation time, over all evacuation algorithms. Note that after a starting edge and a starting point on that edge are fixed, the performance (evacuation cost) of an evacuation algorithm is quantified as the worst-case evacuation time over all adversarial placements of the exit. To that end, we are after algorithms that minimize the evacuation cost in the underlying algorithmic/adversarial choice of a starting point, as a function of the edge lengths $a\geq b\geq c$.

\begin{center}
\begin{tabular}{|c|c|c|lllll}
\hline
Choose Edge & Choose Starting Point & Optimal Evacuation Cost \\
\hline
\hline
Algorithm &	Algorithm & $\ll$ \\
\hline
Adversary & Algorithm & $\lu$ \\
\hline
Algorithm & Adversary & $\ul$ \\
\hline
Adversary & Adversary & $\uu$ \\
\hline
\end{tabular}
\end{center}


Our contributions read as follows. 
\begin{theorem}
\label{thm: four problems}
For triangle $ABC$ with edges $a\geq b \geq c$, we have that 
\begin{align*}
\ll & = \frac{a+b+c}2,\\
\frac{a+b+c}2 \leq 
\lu & \leq \min\left\{ b+c, \frac{(a+b)^2-c^2+4\sqrt3 \tau}{4b}  \right\},  \\
\min\left\{
\frac{1}{2} \sqrt{\frac{2 b^2 (a-c)-(a-2 c) (a+c)^2}{a}} +b,
a+c
\right\}
\leq
\ul & \leq \min\left\{
\frac{a}{2} + b + \frac{c}{2}, 
a+ c
\right\}
,\\
a+\frac{1}{2} \sqrt{\frac{2 a^2 (b-c)-(b-2 c) (b+c)^2}{b}}
\leq
\uu & \leq a+ \frac{b+c}2,
 \end{align*}
where $\tau$ is the area of the triangle. 
\end{theorem}

Theorem~\ref{thm: four problems} is a generalization of some of the results in \cite{CzyzowiczKKNOS15} about the equilateral triangle evacuation problem. Indeed, if the searchers' starting point is $x$ away from the middle point of an equilateral triangle with edges $a=b=c=1$, the authors of~\cite{CzyzowiczKKNOS15} proved that the optimal evacuation cost equals $3/2+x$, showing this way that $
\underline{\mathcal L}(a,a,a)=\overline{\mathcal L}(a,a,a)=3a/2$ and $\underline{\mathcal U}(a,a,a)=\overline{\mathcal U}(a,a,a)=2a$. 
The claims of Theorem~\ref{thm: four problems} imply the exact same bounds by setting $a=b=c$.

\ignore{
Not surprisingly, a plain vanilla algorithm has been shown to be optimal for a list of 2-agent evacuation problems in the wireless model, e.g. when searching the circle, the square, or the equilateral triangle \todo{add references}. In this algorithm that we call \os, the two agents are placed at a point on the perimeter of the search domain, and start moving at full speed in opposite directions until the entire perimeter of the search domain is explored. When the exit is found by any of the agents, and since the information reaches the other agent instantaneously, the latter agent moves to the exit along the shortest line segment. In this work we analyze the same algorithm when the search domain is an arbitrary non-obtuse triangle. 
}

\section{Preliminaries for our Upper Bound Arguments}

In this section we introduce some notation, we recall some useful past results, as well as we make some technical observations that we use repeatedly in our arguments. 
Line segments with endpoints $A,B$ are denoted as $AB$. By abusing notation, and whenever it is clear from the context, we also use $AB$ to refer to the length of the line segment. 
We use $\angle BAC$ to refer to the angle of vertex $A$ of triangle $ABC$, or simply as $\angle A$ (or even more simply $A$), whenever it is clear from the context. 

Our upper bounds for all 4 algorithmic problems we consider are obtained by analyzing the plain-vanilla algorithm that has been proved to be optimal for other search domains in the wireless model, e.g. the circle and the equilateral triangle. We refer to this algorithm as \os.

\textbf{The \os~Algorithm:}
 The algorithm is executed once the two agents are placed at any point on the perimeter of the given geometric search domain. The two agents search the perimeter of the domain in opposite directions at the same unit speed until one of them reports the exit. Then the non-finder, who is notified instantaneously, goes to the exit along the shortest line segment. Consequently, the evacuation cost, for any exit placement, is the time that the non-finder reaches the exit.  Note also that due to the symmetry of the communication model, it does not matter who is the exit finder, rather only when the exit is found. 

Our upper bounds rely on a tight analysis of the \os~algorithm. 
Our proofs use a monotonicity criterion of~\cite{czyzowicz2020priority123} that we present next. The lemma refers to arbitrary unit-speed trajectories $S(t), R(t)\in \reals^2$, $t\geq 0$, assumed to be continuous and differentiable at $t=0$. Let also $S=S(0)$ and $R=R(0)$. We think of $S,R$ as the locations of the two robots where one of them reports the exit. 
We~define the \emph{critical angle} of robot $S(t)$ (similarly of robot $R(t)$) as the angle that the velocity vector $S'(0)$ forms with segment $SR$ (and as the angle that velocity vector $R'(0)$ forms with segment $SR$, for robot $R(t)$). The following lemma characterizes the monotonicity of the evacuation time, assuming that the exit is found in one of the two points $S,R$.

\begin{lemma}[Theorem 2.6 in~\cite{czyzowicz2020priority123}]
\label{lem: monotonicity}
Let $\phi, \theta$ denote the critical angles of robots $S(t), R(t)$ at time $t=0$. Then, the evacuation cost, assuming that the exit is at $S$ or $R$, is \\
strictly increasing if $\coss{\phi}+\coss{\theta}<1$ \\
strictly decreasing if $\coss{\phi}+\coss{\theta}>1$ \\
and constant otherwise. 
\end{lemma}

Intuitively, the proof of Lemma~\ref{lem: monotonicity} relies on that the rate by which the two agents approach each other is $\coss{\phi}+\coss{\theta}$, while the rate by which the time goes by is 1. 
Lastly, our arguments often reduce to the solutions of Non-Linear Programs (NLPs). Whenever we do not provide a theoretical solution to the NLPs, we solve them using computer assisted symbolic (non-numerical) calculations on \textsc{Mathematica} which is guaranteed to report the global optimal value~\cite{referencewolfram2021minimize}.

\subsection{Evacuation Cost for two Line Segments}
\label{sec: special conf}

In this section we analyze the performance of \os~under a special scenario, see Figure~\ref{fig: TwoSegments}. 
We consider a configuration according to which a portion of the search domain has already been explored, and currently the two agents reside at points $B,C$, both moving towards point $A$ at unit speed. Conditioning that the exit is somewhere on segments $AB,AC$, we calculate the evacuation cost, e.g. the worst-case time it takes the last agent to reach the exit, over all placements of the exit. The section is devoted into proving the following lemma.

\begin{lemma}
\label{lem: two segments}
Consider two unit speed robots starting at points $B,C$ simultaneously, and moving toward point $A$ along the corresponding line segments. Assuming that an exit is placed anywhere on segments $AB,AC$, and $\angle A\leq \pi/2$, then the worst case evacuation time equals
$
\max\{AB,AC,BC\}
$.
\end{lemma}

Without loss of generality we assume that in triangle $ABC$ we have $\angle B\geq \angle C$, and hence $AC \geq AB$. 
\begin{figure}
\centering     
\subfigure[Two agents start from points $B,C$ and move towards $A$. The exit is placed somewhere on line segments $AB,AC$. Their critical angles, after searching for time $x$ are denoted as $\phi_x, \theta_x$, respectively.]{\label{fig: TwoSegments}\includegraphics[width=55mm]{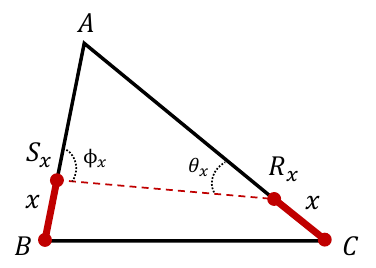}}
~~~~~~~\subfigure[The starting point $S$ of the two robots in the analysis of Section~\ref{sec: analysis for special initial points}.]{\label{fig: SpecialStartingPoint}\includegraphics[width=50mm]{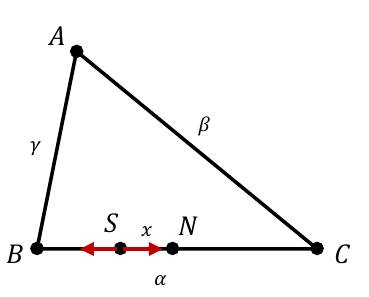}}
\\
\subfigure[The positions of the two robots in Configuration 1. Both $R_1,R_2$ move towards $C$.]{\label{fig: SpecialStartingPointConf1}\includegraphics[width=50mm]{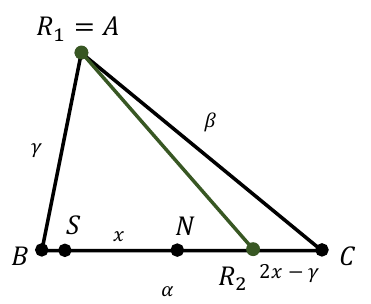}}
~~~~~~~
\subfigure[The positions of the two robots in Configuration 2. Both $R_1,R_2$ move towards $A$.]{\label{fig: SpecialStartingPointConf2}\includegraphics[width=60mm]{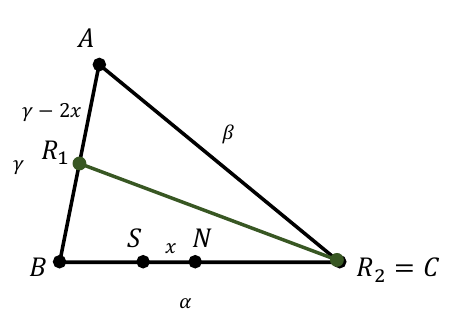}}
\caption{Different configurations of searchers' starting points.}
\end{figure}
We refer to robots at points $B,C$ as $S,R$, respectively. 
Since robots are moving from points $B,C$ towards $A$ at unit speed, it follows that the exit will be found either at some time $x \leq AB$, by either of the two robots, or at time $x\geq AB$. 

\begin{claim}
\label{clm: x>AB}
If the exit is found when $x\geq AB$, then the worst case evacuation cost is $AC$. 
\end{claim}
\begin{proof}
Since the exit is found at $x\geq AB$, the exit will be on edge $AC$. At time $AB$, robot $S$ reaches $A$, robot $R$ still moves towards $A$, and the two robots move towards each other. This instance is covered by Lemma~\ref{lem: monotonicity}. Note that their corresponding critical angles are both $0$ (until they meet), and since $2\coss{0}=2>1$, it follows that the worst placement of the exit is any of the locations of the two robots when $S$ reaches $A$. That gives evacuation cost $AC$. 
\qed \end{proof}

Next we analyze the case that the exit is found at time $x \leq AB$. Our goal is to show that also in this case, the evacuation cost equals 
$
\max\{AB,AC,BC\}
$.
Let $\phi_x, \theta_x$ be robots' critical angles at time $x$, see Figure~\ref{fig: TwoSegments}. Let also $S_x,R_x$ denote their locations on $AB,AC$ respectively. In order to identify the time $x$ that induces the worst case evacuation time, and as per Lemma~\ref{lem: monotonicity}, we need to examine expression 
$$
f(x):=
\coss{\theta_x}+\coss{\phi_x}.
$$

Consider also triangle $AS_xR_x$. Clearly, $A + \phi_x +\theta_x=\pi$, and hence
$f(x)=\coss{\theta_x}+\coss{\pi-(A + \phi_x)}=
\coss{\theta_x}-\coss{A + \phi_x}$.

\begin{claim}
\label{clm: a<pi/3}
If $\angle A <\pi/3$, then $f(x)\neq 1$.
\end{claim}
\begin{proof}
Using trigonometric identities, we have
$$
f(x) = \coss{\theta_x} - \coss{A}\coss{\theta_x}+\sinn{A}\sinn{\theta_x}.
$$
Now consider transformation $y=\coss{\theta_x}$, and therefore $\sqrt{1-y^2}=\sinn{\theta_x}$. Under the previous transformation, we can write that 
$$
f(x)=\left( 1 - \coss{A}\right)y+\sinn{A}\sqrt{1-y^2}.
$$
It follows that equation $f(x)=1$ can be reduced to finding the roots to a degree 2 polynomial in $y$, with discriminant $1-2\coss{A}$. When $A<\pi/3$ the discriminant is negative, and hence $f(x)=1$  has no roots. 
\qed \end{proof}

Note that Claim~\ref{clm: a<pi/3} implies that $f(x)-1$ preserves sign. Moreover for $x=AB$ we have that $\theta_x=0$ and $\phi_x=\pi-A$. Therefore, 
$f(AB)=\coss{0}+\coss{\pi-A}=1-\coss{A}<1$, and moreover $f(x)<1$ for all $x$. By Lemma~\ref{lem: monotonicity}, it follows that the evacuation cost remains increasing for $x\leq AB$, and hence the worst placement of the exit is at $A$, inducing evacuation cost $AC$. At the same time, if $\angle A <\pi/3$ (in the case we are examining) and since $\angle B \geq \angle C$, edge $AC$ is indeed the largest edge. This concludes the lemma in case $\angle A <\pi/3$.

We are now focusing on the case that $\angle A \geq \pi/3$, and we show that the worst placement of the exit is either at $B$ or at $A$. 
Indeed, first we claim that initially, the evacuation cost is decreasing.
\begin{claim}
\label{clm: at x=0 increasing}
The evacuation cost at $x=0$ is decreasing. 
\end{claim}
\begin{proof}
When $x=0$, we have that $\phi_x=B$ and $\theta_x=C$.
But then, we have that
$f(0)=\coss{C}+\coss{B}$.
The previous function, subject to that $B\geq C$ and that $A+B+C=\pi$ (for non-negative variables $A,B,C$) attains a minimum of value $1$.
$f(0)=\coss{C}+\coss{B}\geq 2\coss{C}$. 
By Lemma~\ref{lem: monotonicity}, this shows that at $x=0$ the evacuation cost is decreasing. 
\qed \end{proof}
Now we show that the evacuation cost has only one extreme point with respect to $x\leq AB$ (either maximizer or minimizer). 
\begin{claim}
\label{clm: one critical point}
The evacuation cost attains only one critical point when $x\leq AB$. 
\end{claim}
\begin{proof}
As in the proof of Claim~\ref{clm: a<pi/3}, we attempt to find critical points by solving equation $f(x)=1$. As a reminder, the latter equation reduces to finding the roots of a degree 2 equation (in $y=\coss{\theta_x}$)
$$
\left( 1-
\left( 1 - \coss{A}\right)y
\right)^2 
=
\sinn{A}^2\left(1-y^2\right)
$$
with discriminant $1-2\coss{A}$. Now that $A\geq \pi/3$, the degree 2 polynomial has two (possibly double) roots, and these are 
$$
y_{1,2}=\frac12\left(
1\pm\frac{\sqrt{1-2\coss{A}}}{\tann{A/2}}
\right)
$$
We show that only one of them corresponds to a critical point. To see why, recall that $y=\coss{\theta_x}$, and that $0\leq \theta_x \leq C$, and therefore $\coss{C}\leq y \leq 1$. 

On the other hand, in triangle $ABC$, we have $A+B+C=\pi$, and since $B\geq C$, it follows that $C\leq (\pi-A)/2$. Therefore $\coss{C}\geq \coss{(\pi-A)/2}=\sinn{A/2}$. But then, any critical point $y$ must also satisfy $\sinn{A/2} \leq y \leq 1$. 
When $A\geq \pi/3$ it is easy to see that only one of the two roots $y_{1,2}$ is not less than $\sinn{A/2}$, while since $A\leq \pi/2$, the largest of the roots does not exceed 1. 
\qed \end{proof}
To conclude, by Claim~\ref{clm: one critical point}, the evacuation cost has only one critical point. Since by Claim~\ref{clm: at x=0 increasing} the evacuation cost is initially decreasing, the unique critical point must be a minimizer. As a result, the evacuation cost is maximized either at $x=0$ or at $x=AB$. In the former case the cost equals $BC$, while in the latter, the cost equals $AC$. In any case, the worst case evacuation cost for $x\leq AB$ equals $\max\{BC,AB\}$.

Combining the latest statement with Lemma~\ref{clm: x>AB} (that holds for all angles $A$), we conclude that also in the case that $A\geq \pi/3$, the worst case evacuation cost equals $\max\{AB,AC,BC\}$, as promised by Lemma~\ref{lem: two segments}.

\section{Evacuation Cost Analysis for Special Initial Starting Points}
\label{sec: analysis for special initial points}

The purpose of this section is to analyze the evacuation cost of \os~for searching 
triangle $ABC$ with edge lengths $\alpha, \beta, \gamma$. 
In this section only, we assume that the starting point $S$ of the agents lies on edge $BC$ (of length $\alpha$), and that it is $x$ away from the middle point $N$ of $BC$ on the side of $B$, i.e. $BS=\alpha/2-x$ and $NS=x$, where $0\leq x\leq \alpha/2$, see also Figure~\ref{fig: SpecialStartingPoint}. If $T(x)$ denotes the worst-case evacuation time (over all placements of the exit) starting from $S$, we want to calculate
$$
\max_{0\leq x\leq \alpha/2} T(x), ~~\min_{0\leq x\leq \alpha/2} T(x),
$$
that is, the worst and best starting points on segment $NB$ (inducing worst possible and best possible worst-case costs). 

The best starting point can be thought of as an algorithmic choice, whereas the worst starting point can be thought of as an adversarial choice. 
We note here that $T(x)$ denotes the worst-case evacuation cost for the starting point induced by point $S=S(x)$, which is determined by the adversarial placement of the exit anywhere on the perimeter of the triangle. That is, for each placement $I$ of the exit, and each $x$, $T(x)$ denotes the supremum of the evacuation cost of the algorithm that has agents start from $S$ and given that the exit is at $I$, where the supremum is taken over all $I$. 

Even though the starting point $S$, in this section, is restricted to be in segment $BN$, we will do the performance analysis for all cases regarding the relative lengths of $\alpha,\beta,\gamma$. That will allow us in Section~\ref{sec: bounds from anywhere on the perimeter} to deduce the worst possible and best possible worst-case costs starting anywhere on the perimeter of a triangle for algorithm \os.

\subsection{Observations for Two Special Configurations}
\label{sec: two configurations}

Starting from $S$, agents $R_1,R_2$ traverse clockwise and counterclockiwse the perimeter of $ABC$ until the exit is found. Then, the exit's location is reported, and the non-finder is informed instantaneously and evacuates along the shortest path. 
Note that $R_1$ reaches $B$ no later than $R_2$ reaches $C$. 
Also, $R_1$ reaches point $A$ in total time $\alpha/2-x+\gamma$, while $R_2$ reaches $C$ in time $\alpha/2+x$. The next observation follows immediately from Lemma~\ref{lem: monotonicity}.

\begin{observation}
\label{obs: initially increasing}
Until either $A$ is reached by $R_1$ or $C$ is reached by $R_2$, the evacuation cost is increasing. 
\end{observation}

In the analysis below, we denote by $R_1, R_2$ the locations (points on the perimeter of $ABC$) of the two agents, when the first among points $B,C$ is reached by either of the agents. We also abuse notation and we use $R_1, R_2$ to refer also to the identities of the agents. 
Which of the two points $A,C$ is reached first depends on the original starting point $S$. Therefore, all distances of $R_1,R_2$ to any other point is a function of $x$, and so we will write, for example, $R_1R_2(x)$ or $AR_2(x)$ for the segments $R_1R_2$ and $AR_2$, respectively.

\begin{definition}
\label{def: two configurations}
We say that we are in \textit{Configuration 1} (see Figure~\ref{fig: SpecialStartingPointConf1})
if agent $R_1$ reaches point $A$ no later than $R_2$ reaches $C$. 
We say that we are in \textit{Configuration 2} if agent $R_2$ reaches point $C$ no later than $R_1$ reaches $A$ (see Figure~\ref{fig: SpecialStartingPointConf2}).
\end{definition}

\begin{lemma}
\label{lem: cost in two configurations}
In Configuration 1, we have that
$
T(x)=\alpha/2+\gamma-x +\max\{R_1R_2(x), \beta,2x-\gamma\}
$, where $R_1R_2(x)=AR_2(x)$. 
In Configuration 2, we have that
$
T(x)=\alpha/2+x +\max\{R_1R_2(x), \beta,\gamma-2x\}
$, where $R_1R_2(x)=R_1C(x)$. 
\end{lemma}

\begin{proof}
For Configuration 1, we see that $R_1$ reaches point $A$ no later than $R_2$ reaches $C$ if and only $x\geq \gamma/2$. In this case, since also $R_1$ reaches point $A$ time $\alpha/2-x+\gamma$, and $R_2$ starts $x+\alpha/2$ way from $C$, it follows that 
$$
R_2C=SC-SR_2 = x+\alpha/2-(\alpha/2-x+\gamma)=2x-\gamma.
$$
By Observation~\ref{obs: initially increasing}, up to that moment the evacuation cost keeps increasing. 
From that moment on, $R_1,AR_2$ move towards the same point $C$, and hence Lemma~\ref{lem: two segments} applies, according to which the evacuation cost  
becomes
$$
T(x)=\alpha/2+\gamma-x +\max\{AR_2, AM, R_2C\}=
\alpha/2+\gamma-x +\max\{AR_2, \beta,2x-\gamma\},
$$
where $AR_2=AR_2(x)$. Note that in particular in this case, we have that $R_1=A$. 

For Configuration 2, we see that
$R_2$ reaches point $C$ no later than $R_1$ reaches $A$ if and only $x\leq \gamma/2$. 
In that case, $R_2$ reaches $C$ in time $\alpha/2+x$, and hence while $R_1$ was originally $\alpha/2-x+\gamma$ away from $A$, its current distance is
$$
R_1A=\alpha/2-x+\gamma-(\alpha/2+x)=\gamma-2x.
$$
Up to that moment, and by Observation~\ref{obs: initially increasing}, the evacuation cost was increasing. 
From that moment on, $R_1,R_2$ move towards the same point $A$, and hence Lemma~\ref{lem: two segments} applies, according to which the evacuation cost  
becomes
$$
T(x)=\alpha/2+x +\max\{CR_1, AM, R_1A\}=
\alpha/2+x +\max\{CR_1, \beta,\gamma-2x\},
$$
where $CR_1=CR_1(x)$. Finally, note that in this case, we also have $C=R_2$. 
\qed \end{proof}

Note that in both configurations, the worst case evacuation cost $T(x)$ is a function of $R_1R_2(x)$, where either $R_1=A$ or $R_2=C$. 
The following lemma will be used repeatedly later. 

\begin{lemma}
\label{lem: convexity of R1R2}
In both Configurations 1,2, function $R_1R_2(x)$ is convex in its domain. 
\end{lemma}

\begin{proof}
Consider some function $f(x)=\sqrt{x^2+bx+c}$. We note that 
$\frac {d}{d x^2}  \sqrt{x^2+bx+c} = \frac{4c-b^2}{4(x^2+bx+c)^{3/2}}
$, and therefore $f(x)$ is convex if and only if $4c \geq b^2$. We use this observation to show that $R_1R_2(x)$ is indeed convex. 

In Configuration 1, we have that $R_1R_2(x)=AR_2(x)$. In triangle $AR_2C$ we apply the Cosine Law, and we have that 
\begin{align*}
AR_2(x) 
&=
\left( 
\beta^2+(2x-\gamma)^2-2\beta(2x-\gamma)\cos(C)
\right)^{1/2} \\
&=
2\left( 
x^2 -(\beta\cos(C)+\gamma) x + (\beta^2+\gamma^2+2\beta\gamma\cos(C))/4
\right)^{1/2}
\end{align*}
By our first observation, we can show that $AR_2(x)$ is convex by verifying that 
$$
\beta^2+\gamma^2+2\beta\gamma\cos(C) - (\beta\cos(C)+\gamma)^2 = \beta^2 \sin^2(C) \geq 0.
$$

In Configuration 2, we have that $R_2=C$. In triangle $AR_1C$ and by the Cosine Law we have that 
\begin{align*}
CR_1(x) 
&=
\left( 
\beta^2+(\gamma-2x)^2-2\beta(\gamma-2x)\cos(A)
\right)^{1/2} \\
&=
2\left( 
x^2 +(\beta\cos(A)-\gamma) x + (\beta^2+\gamma^2-2\beta\gamma\cos(A))/4
\right)^{1/2}
\end{align*}
Again by our first observation we verify that $CR_1(x)$ by checking that 
$$
(\beta^2+\gamma^2-2\beta\gamma\cos(A)) - (\beta\cos(A)-\gamma)^2 = \beta^2 \sin^2(A)\geq 0. 
$$
\qed \end{proof}

Now recall that Lemma~\ref{lem: cost in two configurations} gave us the evacuation cost $T(x)$ as a function of the starting point $S=S(x)$. In particular $T(x)$ is expressed as the maximum quantity between linear functions (in $x$) and  $R_1R_2(x)$ which is convex by Lemma~\ref{lem: convexity of R1R2}. Therefore we immediately obtain the following. 
\begin{corollary}
\label{cor: convexity of T}
The evacuation cost $T(x)$ is convex in its domain in both configurations. 
\end{corollary}

Combined with Lemma~\ref{lem: convexity of R1R2}, we will also need to evaluate $R_1R_2(x)$ at the critical points $x\in \{0,\alpha/2,\gamma/2\}$. Whether $R_1R_2(x)$ is defined for these values will depend on the given triangle $ABC$. For the purposes of the lemma, we work under the assumptions that all these values are in the domain of $R_1R_2(x)$ and we will make the proper check when we invoke the lemma. 

\begin{lemma}
\label{lem: some values of R1R2}
In Configuration 1, we have 
\begin{align*}
R_1R_2(\gamma/2) & = \beta \\
R_1R_2(\alpha/2) & = \left(
\frac{\gamma  (\beta +\gamma - \alpha) (\alpha +\beta -\gamma )}{\alpha }
\right)^{1/2}\\
\end{align*}
In Configuration 2, we have 
\begin{align*}
R_1R_2(0) & = \alpha \\
R_1R_2(\gamma/2) & = \beta \\
R_1R_2(\alpha/2) & = \left(
\frac{\alpha  (\beta +\gamma - \alpha) (\alpha +\beta -\gamma )}{\gamma}
\right)^{1/2}\\
\end{align*}
\end{lemma}

\begin{proof}
In Configuration 1, when $x=\gamma/2$, then $R_2=C$, and therefore $R_1R_2(\gamma/2)=AC=\beta$. 

Next we study $x=\alpha/2$ of the Configuration 1. 
Note that $R_1=A$ and that $R_2C=\alpha-\gamma$. 
Using the Cosine Law in triangle $AR_2C$, we have 
\begin{align*}
AR_2(\alpha/2) 
&=
\left( 
\beta^2+(\alpha-\gamma)^2-2\beta(\alpha-\gamma)\cos(C)
\right)^{1/2} \\
&
=
\left(
\frac{\gamma  (\beta +\gamma - \alpha) (\alpha +\beta -\gamma )}{\alpha }
\right)^{1/2},
\end{align*}
where the last equality is by expressing $\cos(C)$ as a function of the triangle edges using the Cosine Law in triangle $ABC$ according to which $\gamma^2=\alpha^2+\beta^2-2\alpha\beta\cos(C)$.

We move on with Configuration 2. 
When $x=0$, then $R_1=B$, and therefore $R_1R_2(0)=BC=\alpha$. 
Also, when $x=\gamma/2$, then $R_1=A$, and therefore $R_1R_2(\gamma/2)=AC=\beta$. 

Finally, when $x=\alpha/2$, then $AR_1=\gamma-\alpha$. By the Cosine Law in $CR_1A$, we have that 
\begin{align*}
CR_1(\alpha/2) 
&=
\left( 
\beta^2+(\gamma-\alpha)^2-2\beta(\gamma-\alpha)\cos(A)
\right)^{1/2} \\
&
=
\left(
\frac{\alpha  (\beta +\gamma - \alpha) (\alpha +\beta -\gamma )}{\gamma}
\right)^{1/2},
\end{align*}
where the last equality is by expressing $\cos(A)$ as a function of the triangle edges using the Cosine Law in triangle $ABC$ according to which $\alpha^2=\beta^2+\gamma^2-2\beta\gamma\cos(A)$.
\qed \end{proof}

We also use repeatedly the polytope
$\Delta\subseteq \reals^3$ identified by all possible triples of edge lengths, i.e. $(\alpha,\beta,\gamma)\in \Delta$ if and only if 
\begin{align*}
&\alpha+\beta \geq \gamma \\
&\alpha+\gamma \geq \beta \\
&\beta+\gamma \geq \alpha \\
&\alpha,\beta,\gamma\geq 0
\end{align*}

We are now ready to provide tight bounds of \os~for quantities $
\max_{0\leq x\leq \alpha/2} T(x)$ and $\min_{0\leq x\leq \alpha/2} T(x),
$ where $x$, and the starting point of the two robots as in Figure~\ref{sec: analysis for special initial points}. Our analysis is done by cases depending on the size of the edges, that is 
Section~\ref{sec: abc} analyzes the case $\alpha \geq \beta \geq \gamma$,
Section~\ref{sec: acb} analyzes the case $\alpha \geq \gamma \geq \beta$,
Section~\ref{sec: bac} analyzes the case $\beta \geq \alpha \geq \gamma$,
Section~\ref{sec: bca} analyzes the case $\beta \geq \gamma \geq \alpha$,
Section~\ref{sec: cab} analyzes the case $\gamma \geq \alpha \geq \beta$
and
Section~\ref{sec: cba} analyzes the case $\gamma \geq \beta \geq \alpha$.
In each case we further analyze the cost of $T(x)$ in Configurations 1 and 2, when applicable.


\subsection{The case $\alpha \geq \beta \geq \gamma$}
\label{sec: abc}
In Configuration 1 we have 
$x\geq \gamma/2$ and therefore
$$
T(x)=\alpha/2+\gamma-x +\max\{AR_2, \beta,2x-\gamma\}
=\alpha/2+\gamma-x +\max\{AR_2, \beta\},
$$
because $2x-\gamma \leq \alpha - \gamma \leq \beta$, by the triangle inequality. 
We need to find the extreme values of $T(x)$ in the interval $x\in [\gamma/2,\alpha/2]$.

\begin{lemma}
\label{lem: abc conf1}
Given that $\alpha \geq \beta \geq \gamma$, we have that 
\begin{align*}
\max_{\gamma/2 \leq x\leq \alpha/2} T(x) & = \alpha/2 + \beta + \gamma/2 ~~(\textrm{attained at}~x=\gamma/2)\\
\min_{\gamma/2\leq x\leq \alpha/2} T(x) & = \beta+\gamma
~~(\textrm{attained at}~x=\alpha/2).
\end{align*}
\end{lemma}

\begin{proof}
In the current configuration we have that 
$$
T(x)=\alpha/2+\gamma +\max\{AR_2(x) - x, \beta - x\}
$$
We observe that as $x$ ranges between $\gamma/2$ and $\alpha/2$, point $R_2$ ranges in segment $BC$. As a result, $AR_2(x)$ ranges between $\gamma$ and $\beta$, and hence, for all $x\in [\gamma/2,\alpha/2]$, we have $\max\{AR_2(x), \beta \}=\beta$. But then
$
T(x)=\alpha/2+\gamma +\beta - x,
$
attaining its minimum value at $x=\alpha/2$, and its maximum value at $x=\gamma/2$. 
\qed \end{proof}

Next we move to Configuration 2, in which case $x\leq \gamma/2$ and therefore
$$
T(x)=\alpha/2+x +\max\{CR_1, \beta,\gamma-2x\}
=\alpha/2+x +\max\{CR_1, \beta\}
$$
We show the next lemma. 

\begin{lemma}
\label{lem: abc conf2}
Given that $\alpha \geq \beta \geq \gamma$, we have that 
\begin{align*}
\max_{0\leq x\leq \gamma/2} T(x) & = \alpha/2 + \max \{\alpha, \beta + \gamma/2\} ~~(\textrm{attained at}~x=0~\textrm{or}~x=\gamma/2~resp)\\
\min_{0\leq x\leq \gamma/2} T(x) & = \frac{\alpha^2-\beta^2+\alpha \gamma + 2\beta \gamma }{2\gamma}~~(\textrm{attained at}~x=\gamma/2-\cos(A)\beta),
\end{align*}
where $\cos(A)$ can be computed with the Cosine Law in triangle $ABC$. 
\end{lemma}

\begin{proof}
In the current configuration we have that 
$$
T(x)=\alpha/2+\max\{x+CR_1(x), \beta+x\},
$$
with domain $[0,\gamma/2]$. By Corollary~\ref{cor: convexity of T}, $T(x)$ is convex, and hence attains its maximum values in one of the endpoints $\{0,\gamma/2\}$. Using Lemma~\ref{lem: some values of R1R2}, we calculate $T(x)$ at these two values, from which we conclude that $\max_{0\leq x\leq \gamma/2} T(x) = \alpha/2 + \max \{\alpha, \beta + \gamma/2\}$. 

The minimum of $T(x)$ in the same domain is assumed when $CR_1(x)=\beta$, for $x_0\neq 0$. But then triangle $AR_1C$ is isosceles, e.g. when $R_1C=AC$. In this case, $\cos(A)=(\gamma-2x_0)/(2\beta)$, and therefore $x_0=\gamma/2-\cos(A)\beta$ (which can be verified to be non-negative independently). But then, 
$$
T(x_0)= \alpha/2+\beta+\gamma/2-\cos(A)\beta,
$$
where $\cos(A)=(\beta^2+\gamma^2-\alpha^2)/(2\beta\gamma)$ is obtained by the Cosine Law in triangle $ABC$. 
\qed \end{proof}

\subsection{The case $\alpha \geq \gamma \geq \beta$}
\label{sec: acb}

In Configuration 1 we have $x\geq \gamma/2$ and so
$$T(x)
=\alpha/2+\gamma-x +\max\{AR_2, \beta,2x-\gamma\}
=\alpha/2+\gamma-x +\max\{AR_2, \beta\}
$$
since $2x-\gamma \leq \alpha - \gamma \leq \beta$, where the last inequality is due to triangle inequality. 

We show the next lemma. 

\begin{lemma}
\label{lem: acb conf1}
Given that $\alpha \geq \gamma \geq \beta$, we have that 
\begin{align*}
\max_{\gamma/2 \leq x\leq \alpha/2} T(x) & = \alpha/2 + \beta + \gamma/2~~(\textrm{attained at}~x=\gamma/2)\\
\min_{\gamma/2\leq x\leq \alpha/2} T(x) & = \beta+\gamma
~~(\textrm{attained at}~x=\alpha/2).
\end{align*}
\end{lemma}

\begin{proof}
In the current configuration we have that 
$$
T(x)=\alpha/2+\gamma -x + \max\{AR_2(x), \beta\},
$$
with domain $[\gamma/2,\alpha/2]$. 
For this we show that for all $x\in [\gamma/2,\alpha/2]$, we have that $AR_2(x)\leq  \beta$. 

Indeed, by Lemma~\ref{lem: some values of R1R2}, we know that $AR_2(\gamma/2)=\beta$, and by Lemma~\ref{lem: convexity of R1R2} that $AR_2(x)$ is convex. Moreover, it is straightforward that $AR_2(x)$ is strictly decreasing at $x=\gamma/2$. We show that there is no $x_0 \in (\gamma/2,\alpha/2)$ with $AR_2(x_0) = \beta$, which would imply that $AR_2(x)\leq  \beta$.

We observe that if $AR_2(x_0)= \beta$, then triangle $AR_2C$ is isosceles, with base $2x_0-\gamma$, and angle of the base vertices equal to $C$. Therefore, $\cos(C)=(2x_0-\gamma)/(2\beta)$, from which we get $x_0 = \gamma/2+\cos(C) \beta$. We argue that $x_0\geq \alpha/2$. Indeed, algebraic calculations show that 
$$
\gamma/2+\cos(C) \beta - \alpha/2=
\frac{\gamma (\alpha-\gamma)+\beta^2}{2 \alpha}\geq 0,
$$
since $\alpha \geq \gamma$. 
\ignore{
For this we consider the nonlinear program with objective $x_0 - \alpha/2$ which reads 
\begin{align*}
\min&~~ \gamma/2+\cos(C) \beta - \alpha/2
\\
s.t. ~~& \alpha \geq \gamma \geq \beta \\
& (\alpha,\beta,\gamma) \in \Delta,
\end{align*}
where $\cos(C)=\frac{\alpha ^2+\beta ^2-\gamma ^2}{2 \alpha  \beta }$ is obtained by the Cosine Law in $ABC$. 
Without loss of generality, we may assume that $\alpha=1$, and the solution to the NLP above is given by $\beta=1, \gamma=0$ with the objective attaining the value $0$. This shows that $x_0 \geq \alpha/2$ as claimed. 
}
Overall, we showed that when $\alpha\geq \gamma \geq \beta$, we have that 
$
T(x)
=\alpha/2+\gamma -x + \max\{AR_2(x), \beta\}
=\alpha/2+\gamma -x + \beta,
$
and therefore the maximum is attained at $x=\gamma/2$, and the minimum when $x=\alpha/2$. 
\qed \end{proof}

Next we study Configuration 2, in which $x\leq \gamma/2$ and therefore
$$
T(x)=\alpha/2+x +\max\{CR_1, \beta,\gamma-2x\}.
$$
We need to find the extreme values of $T(x)$ in the interval $x\in [0,\gamma/2]$.

\begin{lemma}
\label{lem: acb conf2}
Given that $\alpha \geq \gamma \geq \beta$, we have that 
\begin{align*}
\max_{0\leq x\leq \gamma/2} T(x) & = 3\alpha/2 ~~(\textrm{attained at}~x=0)\\
\min_{0\leq x\leq \gamma/2} T(x) & = \frac{(\alpha+\gamma)^2-\beta^2+4\sqrt{3} \tau}{4\gamma}~~(\textrm{attained at}~x=
\gamma/2-\beta/2 \left(\cos(A) +\sqrt{3}\sin{A}\right)   ),
\end{align*}
where $\tau$ is the area of triangle $ABC$, and $\cos(A), \sin(A)$ can be given by the Cosine Law as functions of $\alpha,\beta,\gamma$.  
\end{lemma}

\begin{proof}
By Corollary~\ref{cor: convexity of T}, $T(x)$ is convex, so it attains its maximum either at $x=0$ or at $x=\gamma/2$. Using Lemma~\ref{lem: some values of R1R2}, we calculate
$$
T(0)
=\alpha/2+\max\{CR_1(0), \beta,\gamma\}
=\alpha/2+\max\{\alpha, \beta,\gamma\}
3\alpha/2,
$$
and 
$$
T(\gamma/2)
=\alpha/2+\gamma/2 +\max\{CR_1(\gamma/2), \beta,0\}
=\alpha/2+\beta.
$$
But then, 
$$
\max_{0\leq x\leq \gamma/2} T(x)
=\max\{
T(0), T(\gamma/2)
\}
=
\max\{
3\alpha/2,
\alpha/2+\beta
\}
=3\alpha/2.
$$
and the maximum is attained at $x=0$. 

Next we calculate $\min_{0\leq x\leq \gamma/2} T(x)$, by writing $T(x)$ as a piecewise function. First by comparing $\beta,\gamma-2x$ it is easy to see that 
$$
T(x)=\alpha/2+x +
\left\{
\begin{array}{ll}
\max\{CR_1(x),\gamma-2x\} &,~\textrm{if}~0\leq x \leq \gamma/2-\beta/2 \\
\max\{CR_1(x), \beta\} &,~\textrm{if}~\gamma/2-\beta/2\leq x \leq \gamma/2
\end{array}
\right.
$$

We first claim that if $x \leq \gamma/2-\beta/2$, then $CR_1(x)\geq \gamma-2x$. 
To see why, we compute $CR_1(x)$ using the Cosine Law in triangle $AR_1C$, according to which
$
CR_1(x)= \left(
\beta^2+(\gamma-2x)^2-2\beta(\gamma-2x)\cos(A)
\right)^{1/2}.
$
But then equating $CR_1(x)$ with $\gamma-2x$ gives $\bar x= \gamma/2-\beta/(4\cos(A))$. We show that $\bar x\leq 0$, and that would imply that $CR_1(x)\geq \gamma-2x$. For this we note that 
$$\cos(A) \gamma - \beta/2 = \frac{\gamma^2-\alpha^2}{2 \beta}\leq 0,$$
since $\alpha \geq \gamma$. 
\ignore{
consider the nonlinear program
\begin{align*}
\max&~~ \cos(A) \gamma - \beta/2
\\
s.t. ~~& \alpha \geq \gamma \geq \beta \\
& (\alpha,\beta,\gamma) \in \Delta,
\end{align*}
where $\cos(A)=(\beta^2+\gamma^2-\alpha^2)/(2\beta\gamma)$ is obtained by the Cosine Law in triangle $ABC$. 
Condition on that $\alpha=1$ (w.l.o.g), the solution to the nonlinear program is given by $\alpha=\beta=1$ and $\gamma=1/2$, and the maximum value is $0$.
} 
But then, $\cos(A) \gamma \leq \beta/2$ for all triangles for which $\alpha\geq \gamma\geq \beta$, and so $\bar x = \gamma/2-\beta/(4\cos(A)) \leq 0$. 
Our argument shows that $CR_1(x)\neq \gamma-2x$ when $0\leq x \leq \gamma/2-\beta/2$ but also by Lemma~\ref{lem: some values of R1R2} we have $CR_1(0)=\alpha \geq \gamma$. By continuity, we conclude that $CR_1(x)\geq \gamma-2x$ as promised (when $x\leq leq \gamma/2-\beta/2$). 

Next we study $T(x)$ in the interval $[ \gamma/2-\beta/2, \gamma/2]$, for which we need to compare $CR_1(x), \beta$. Note that by Lemma~\ref{lem: some values of R1R2}, we have $CR_1(\gamma/2), \beta$. Any $x_0<\gamma/2$ that would make $CR_1(x_0)= \beta$ would induce a isosceles $CAR_1$, with base equal to $\gamma - 2x_0$, and base angle $A$. But then, $\cos(A)=(\gamma-2x_0)/(2\beta)$, which implies that $x_0 :=\gamma/2 - \cos(A) \beta\geq \gamma/2 - \beta/2$ (since $A\geq \pi/3$ is the largest angle in $ABC$). It is also easy to see that $CR_1(x)$ is increasing at $x=\gamma/2$, and since $CR_1(\gamma/2)=\beta$ and the convexity of $CR_1(x)$ (Lemma~\ref{lem: convexity of R1R2}), we conclude that $T(x) = \alpha/2+x+CR_1(x)$ if $\gamma/2-\beta/2 \leq x \leq x_0$, and $T(x)= \alpha/2+x+\beta$ if $x_0 \leq x \leq \gamma/2$. 

To conclude, we have 
$$
T(x)=
\left\{
\begin{array}{ll}
\alpha/2+x + CR_1(x) &,~\textrm{if}~0\leq x \leq x_0 \\
\alpha/2+x +\beta &,~\textrm{if}~x_0 <\leq x \gamma/2
\end{array}
\right.
$$
Now we show that $T(x)$ is minimizes in the interval $[0,x_0]$ and we show how to find the minimizer. 

By Corollary~\ref{cor: convexity of T} we know that $T(x)$ is convex. We are showing that $T(x)$ is decreasing at $x\rightarrow 0^+$, and increasing at $x\rightarrow x_0$. For this, we calculate
$$
\frac{d}{dx} T(x)
=
\frac{2 \beta  \cos (A)-2 \gamma +4 x}{\sqrt{-2 \beta  \cos (A) (\gamma -2 x)+\beta ^2+(\gamma -2 x)^2}}+1,
$$
where in particular $\cos(A)$ is given explicitly as a function of $\alpha,\beta,\gamma$. We denote the last expression in $x$ by $f(x)$. Basic calculations then derive that 
$$
\lim_{x\rightarrow 0^+} T'(x) = f(0)=\frac{-\alpha ^2+\alpha  \gamma +\beta ^2-\gamma ^2}{\alpha  \gamma }
$$
and 
$$
\lim_{x\rightarrow x_0^-} T'(x) = f(x_0)=\frac{\alpha ^2-\beta ^2+\beta  \gamma -\gamma ^2}{\beta  \gamma },
$$
where both $f(0), f(x_0)$ are expressions on $\alpha,\beta,\gamma$. 
The nonlinear program then,
\begin{align*}
\max&~~ f(0)
\\
s.t. ~~& \alpha \geq \gamma \geq \beta \\
& (\alpha,\beta,\gamma) \in \Delta,
\end{align*}
 with the condition that $\alpha=1$ attains the optimizer $\alpha=\beta=\gamma=1$, and its value becomes $0$, that is over all triangles, $f(0)\leq 0$, and hence $T(x)$ is initially decreasing. 
 
 Similarly, we consider the nonlinear program 
\begin{align*}
\min&~~ f(x_0)
\\
s.t. ~~& \alpha \geq \gamma \geq \beta \\
& (\alpha,\beta,\gamma) \in \Delta.
\end{align*}
With the condition that $\alpha=1$, the minimizer is again 
$\alpha=\beta=\gamma=1$, and its value becomes $0$. This shows that over all triangles, $f(x_0)\geq 0$, and hence $T(x)$ is eventually increasing in $[0,x_0]$. Recall also that $T(x)=\alpha/2+x +\beta$, when $x\geq x_0$, that is $T(x)$ is increasing in $x\geq x_0$. Overall, this shows that $T(x)$ attains a unique minimum in $[0,\gamma/2]$, and that is the unique minimum of the convex function 
$$
g(x) := \alpha/2+x + CR_1(x)
=\alpha/2+x +  \left(
\beta^2+(\gamma-2x)^2-2\beta(\gamma-2x)\cos(A)
\right)^{1/2}
$$
in the interval $[0,\gamma/2]$. 

For this, we solve $g'(x)=0$ for $x$, and we find that 
$$
x_{1,2} =  \left(3 \gamma - 3 \beta \cos(A) \mp \sqrt{3} \beta \sin(A)\right)/6. 
$$
It is easy to see that for all triangles with $\alpha\geq \gamma\geq \beta$ we have that $x_2\geq \gamma/2$, and so the only root of $g'(x)$, hence the minimizer of $T(x)$ is $ x_1 = \frac16 (3 \gamma - 3 \beta \cos(A) - \sqrt{3} \beta \sin(A))$. It follows that 
$\min_{0\leq x \leq \gamma} T(x) = T(x_1)$ which simplifies to 
$$
\frac{\sqrt{3} \sqrt{-((\alpha -\beta -\gamma ) (\alpha +\beta -\gamma ) (\alpha -\beta +\gamma ) (\alpha +\beta +\gamma ))}+(\alpha +\gamma )^2-\beta ^2}{4 \gamma }
$$
Recalling Heron's formula $\sqrt{p(p-\alpha)(p-\beta)(p-\gamma)}$ for the area $\tau$ of a triangle with half perimeter $p$, gives the promised lower bound. 
\qed \end{proof}

\subsection{The case $\beta \geq \alpha \geq  \gamma$}
\label{sec: bac}

First we study Configuration 1 in which $x\geq \gamma/2$ and so
$$T(x)=\alpha/2+\gamma-x +\max\{AR_2, \beta,2x-\gamma\}
=\alpha/2+\gamma-x +\max\{AR_2, \beta\},
$$
because $2x-\gamma\leq \alpha-\gamma \leq \beta$, by the triangle inequality.

We show the next lemma. 

\begin{lemma}
\label{lem: bac conf1}
Given that $\beta \geq \alpha \geq \gamma$, we have that 
\begin{align*}
\max_{\gamma/2 \leq x\leq \alpha/2} T(x) & = \alpha/2+\beta+\gamma/2~~(\textrm{attained at}~x=\gamma/2)\\
\min_{\gamma/2\leq x\leq \alpha/2} T(x) & = \beta+\gamma
~~(\textrm{attained at}~x=\alpha/2).
\end{align*}
\end{lemma}

\begin{proof}
Recall that 
$T(x)=\alpha/2+\gamma-x +\max\{AR_2, \beta\}
$. As starting point $S$ ranges on $BC$, we have $AR_2(x) \leq \max\{\beta,\gamma\}$, and therefore 
$T(x)=\alpha/2+\gamma-x +\beta$
attaining its minimum at $x=\alpha/2$ and its maximum at $x=\gamma/2$. 
\qed \end{proof}

Second, we study Configuration 2 in which $x\leq \gamma/2$ and therefore
$$
T(x)=\alpha/2+x +\max\{CR_1, \beta,\gamma-2x\}
=\alpha/2+x +\max\{CR_1, \beta\},
$$
because $\beta \geq \gamma$. 

\begin{lemma}
\label{lem: bac conf2}
Given that $\beta \geq \alpha \geq \gamma$, we have that 
\begin{align*}
\max_{0\leq x\leq \gamma/2} T(x) & =  \alpha/2+\beta + \gamma/2 ~~(\textrm{attained at}~x=\gamma/2)\\
\min_{0\leq x\leq \gamma/2} T(x) & = \alpha/2+\beta~~(\textrm{attained at}~x=0),
\end{align*}
\end{lemma}

\begin{proof}
We observe that for all $x\in [0,\gamma/2]$, we have that $CR_1(x) \leq \max\{\alpha,\beta\}=\beta$, since point $R_1$ ranges in the segment $AB$. 
Therefore
$
T(x) = \alpha/2+x+\beta,
$
and so the maximum and minimum are attained at $x=\gamma/2$ and $x=0$, respectively. 
\qed \end{proof}


\subsection{The case $\beta \geq  \gamma \geq \alpha $}
\label{sec: bca}

Regarding Configuration 1, we recall that $x\leq \alpha/2$. Since this configuration happens only if $x\geq \gamma/2$ and $\gamma\geq \alpha$, we conclude that the domain of the function in this case is either empty, or it is covered by Configuration 2 (when $\alpha=\gamma$). 

We now move to Configuration 2. Recall that $x\leq \alpha/2$, and since this configuration happens only when $x\leq \gamma/2$ and $\gamma\geq \alpha$ we need to determine the extreme values of 
$$
T(x)=\alpha/2+x +\max\{CR_1, \beta,\gamma-2x\}
=
\alpha/2+x +\max\{CR_1, \beta\},
$$
in the interval $x\in [0,\alpha/2]$ (since $\beta\geq \gamma$).

\begin{lemma}
\label{lem: bca conf2}
Given that $\beta \geq \gamma \geq \alpha$, we have that 
\begin{align*}
\max_{0\leq x\leq \alpha/2} T(x) & = \alpha+\beta  ~~(\textrm{attained at}~x=\alpha/2)\\
\min_{0\leq x\leq \alpha/2} T(x) & =\alpha/2+\beta ~~(\textrm{attained at}~x=0),
\end{align*}
\end{lemma}

\begin{proof}
We observe that for all $x\in [0,\alpha/2]$, we have that $CR_1(x) \leq \max\{\alpha,\beta\}=\beta$, since point $R_1$ ranges in the segment $AB$. 
Therefore
$
T(x) = \alpha/2+x+\beta,
$
and so the maximum and minimum are attained at $x=\alpha/2$ and $x=0$, respectively. 
\qed \end{proof}

\subsection{The case $\gamma \geq \alpha \geq \beta $}
\label{sec: cab}

As in the previous section, the domain of our function in Configuration 1 is either empty, or it is covered by Configuration 2 (when $\alpha=\gamma$). 
So we move on to the study of Configuration 2. Recall that $x\leq \alpha/2$, and since this configuration happens only when $x\leq \gamma/2$ and $\gamma\geq \alpha$ we need to determine the extreme values of 
$$
T(x)=\alpha/2+x +\max\{CR_1, \beta,\gamma-2x\},
$$
in the interval $x\in [0,\alpha/2]$.

The following lemma will be needed both for the proofs of Lemma~\ref{lem: cab conf2} and Lemma~\ref{lem: cba conf2}. 

\begin{lemma}
\label{lem: T1 is less T2,T2}
Let 
$T_1(x)=\alpha/2+x+CR_1(x)$,
$T_2(x)=\alpha/2+\beta +x$,
$T_3(x)=\alpha/2+\gamma-x$,
and suppose that $\gamma\geq \max\{\alpha,\beta\}$. If in addition we have that $A\leq \pi/3$, then 
$T_1(x) \leq \max\{T_2(x),T_3(x)\}$, 
for all $0\leq x \leq \alpha/2$. 
\end{lemma}

\begin{proof}
We calculate 
$T_2(0)= \alpha/2+\beta$ 
and $T_3(0) = \alpha/2+\gamma$. 
Since $\gamma\geq \max\{\alpha,\beta\}$ we see that $T_3(0) \geq T_2(0)$. 
Moreover,
$T_2(\alpha/2)= \alpha+\beta$ 
and $T_3(\alpha/2) = \gamma$. By the triangle inequality we get that $T_3(\alpha/2) \leq T_2(\alpha_2)$, so that 
$$
\max\{T_2(x),T_3(x)\}
=
\left\{
\begin{array}{ll}
\alpha/2+\gamma - x &~\textrm{, if}~x\leq \gamma/2-\beta/2\\
\alpha/2+\beta + x &~\textrm{, if}~x > \gamma/2-\beta/2
\end{array}
\right..
$$

Using Lemma~\ref{lem: some values of R1R2}, we also calculate $T_1(0)=\alpha/2+x+CR_1(0) = 3\alpha/2 \leq \alpha/2+\gamma = T_3(0)$, because  $\gamma\geq \max\{\alpha,\beta\}$. Next we claim that $T_1(\alpha/2) \leq T_2(\alpha/2)$. 
To see why, note that using the Cosine Law in $CR_1A$ we have that 
$
CR_1(x)=\left(
\beta^2+(\gamma-2x)^2 - 2\beta(\gamma-2x)\cos(A)
\right)^{1/2}
$,
where $\cos(A)= (\beta^2+\gamma^2-\alpha^2)/(2\beta\gamma)$, by the Cosine Law in $ABC$. Then, we solve the nonlinear program 
\begin{align*}
\max&~~ T_1(\alpha/2) - T_2(\alpha/2)
\\
s.t. ~~& \gamma \geq \max\{\beta,\gamma\} \\
& (\alpha,\beta,\gamma) \in \Delta.
\end{align*}
With the condition that $\alpha=1$, the minimizer is attained at again 
$\alpha=\gamma=1, \beta =5/16$, and its value becomes $0$. This shows that$T_1(\alpha/2) - T_2(\alpha/2)$.

Lastly, using also that $A\leq \pi/3$, we show that $T_1(\gamma/2-\beta/2) \leq \alpha/2+\beta/2+\gamma/2$ (were the last expression also equals to $T_2(\gamma/2-\beta/2)$ and to $T_2(\gamma/2-\beta/2)$). Recalling also that 
$A\leq \pi/3$ and hence $\cos(A) \geq 1/2$, we utilize the non linear program 
\begin{align*}
\max&~~ T_1(\gamma/2-\alpha/2) - (\alpha/2-\beta/2-\gamma/2)
\\
s.t. ~~& \gamma \geq \max\{\beta,\alpha\} \\
& \cos(A) \geq 1/2 \\
& (\alpha,\beta,\gamma) \in \Delta.
\end{align*}
Note that in particular $\cos(A)$ is expressed as a function of $\alpha,\beta,\gamma$, as we also may assume that $\alpha=1$. Then, the optimal value equals $0$ and is obtained for $\beta=1/2, \gamma=(1+\sqrt{13})/4\approx 1.15139$, showing that $T_1(\gamma/2-\beta/2) \leq \alpha/2+\beta/2+\gamma/2$. 

Lastly, we observe that $T_2,T_3$ are linear functions. When $x\leq \gamma/2-\beta/2$, then $\max\{T_2(x),T_3(x)\}= T_3(x)$ which as per our argument above, dominates $T_1(x)$ (for all $x\leq \gamma/2-\beta/2$). 
Also when 
$x\geq \gamma/2-\beta/2$, then $\max\{T_2(x),T_3(x)\}= T_2(x)$ which as per our argument above, dominates $T_1(x)$ (for all $x\geq \gamma/2-\beta/2$). 
Our main claim follows. 
\qed \end{proof}

\begin{lemma}
\label{lem: cab conf2}
Given that $\gamma \geq \alpha \geq \beta$, we have that 
\begin{align*}
\max_{0\leq x\leq \alpha/2} T(x) & = \alpha/2 +\max\{
\gamma,\alpha/2+\beta
\}  ~~(\textrm{attained at}~x=0,\alpha/2,~\textrm{resp.})\\
\min_{0\leq x\leq \alpha/2} T(x) & =
\max
\left\{
\frac{\alpha+\beta+\gamma}2,\frac{(\alpha+\gamma)^2-\beta^2+4\sqrt{3} \tau}{4\gamma}
\right\}
\end{align*}
where $\tau$ is the area of triangle $ABC$, and $\cos(A), \sin(A)$ can be given by the Cosine Law as functions of $\alpha,\beta,\gamma$.  
Moreover, the minimum value of $T(x)$ is attained at $x=\gamma/2-\beta/2$ if $A\leq \pi/3$, and at 
$x=\frac{\gamma}2-\frac{\beta}2\left(
\cos(A)+ \frac{\sqrt3}3\sin(A)
\right)$ if $A\geq \pi/3$.  
\end{lemma}

\begin{proof}
First we compute the maximum value of $T(x)$ over $[0,\alpha/2]$.
Our analysis (for the maximum value only) relies on that $\gamma\geq \max\{\alpha,\beta\}$, and hence it also holds for the upper bound claim of 
Lemma~\ref{lem: cba conf2}.

 $T(x)$ is convex by Lemma~\ref{cor: convexity of T}, so its maximum value is attained either at $x=0$ or at $x=\alpha/2$. Using Lemma~\ref{lem: some values of R1R2}, we calculate
$$
T(0)=\alpha/2+\max\{CR_1(0),\beta,\gamma\}
=\alpha/2+\max\{\alpha,\beta,\gamma\}
=\alpha/2+\gamma,
$$
as well as 
$$
T(\alpha/2)=\alpha/2+\alpha/2+\max\{CR_1(\alpha/2),\beta,\gamma-\alpha\}
=\alpha+\max\{CR_1(\alpha/2),\beta\},
$$
where the last equality is due to that $\gamma-\alpha\leq \beta$, by the triangle inequality, and $CR_1(\alpha/2)$ is given by Lemma~\ref{lem: some values of R1R2}.
Next we argue that for all triangles in which $\gamma$ is the dominant edge, we have $CR_1(\alpha/2),\beta$. For this we set up the following nonlinear program
\begin{align*}
\max&~~ CR_1(\alpha/2) -\beta
\\
s.t. ~~& \gamma \geq \max\{\alpha,\beta\} \\
& (\alpha,\beta,\gamma) \in \Delta.
\end{align*}
Without loss of generality we may assume that $\alpha=1$, so that the optimizers of the nonlinear program are of the form $\alpha=\gamma=1$, and $\beta\leq \gamma$, giving maximum value $0$. Hence, $CR_1(\alpha/2)\leq \beta$. 
We conclude that $T(\alpha/2)=\alpha+\beta$, so that 
$
\max_{0\leq x\leq \alpha/2} T(x)
=\alpha/2 + \max\{
\gamma,\alpha/2+\beta
\}
$
attained at $x=0,\alpha/2$, respectively. 

Next we calculate the lower bound. Note that $\alpha$ is the second largest edge, and hence $A$ is the second largest angle of $ABC$. We examine two cases. 

In the first case, we assume that $A\leq \pi/3$. Then,  Lemma~\ref{lem: T1 is less T2,T2} applies, according to which 
$T(x)=\alpha/2+x+\max\{\beta,\gamma-2x\}$. 
Then we observe that $\beta=\gamma-2x$ at $x_0=\gamma/2-\beta/2\leq \alpha/2$ (due to the triangle inequality), and hence $T(x)$ is minimized at $x_0$. It is also easy to see that $T(x_0) = \alpha/2+\beta/2+\gamma/2$. 

In the second case we assume that $A\geq \pi/3$. Let also 
$T_1(x)=\alpha/2+x+CR_1(x)$. We show that the minimum of $T(x)$ over $[0,\alpha/2]$ is the minimum of $T_1(x)$ in the same interval. For this, we also set 
$T_2(x)=\alpha/2+\beta +x$ and
$T_3(x)=\alpha/2+\gamma-x$ (and note that $T(x) = \max\{T_1(x), T_2(x), T_3(x) \}$). 

We compute some important values of $x$. First, we see that $T_1(x)=T_3(x)$ at $x_1:=\frac{(\gamma^2-\alpha^2)\gamma}{2(\beta^2+\gamma^2-\alpha^2)}$ as well as that $T_1(x)=T_2(x)$ at $x_2:=\frac{(\alpha+\beta)(\alpha-\beta)}{2\gamma}$ (we also have $T_1(\gamma/2)=T_2(\gamma/2)$, but this solution is rejected because $x\leq \alpha/2$ and $\gamma\geq \max\alpha$).
Since $T_1(0) \leq T_3(0)$, it follows that $T_1(x)\geq T_3(x)$, for all $x\geq x_1$. Moreover, since $T_1(\alpha/2) \leq T_2(\alpha/2)$, it follows that $T_1(x)\geq T_2(x)$, for all $x\leq x_2$.

The next claim is that $x_1 \leq \gamma/2-\beta/2 \leq x_2$. For this, we consider the nonlinear programs 
\begin{align*}
\max&~~ f_i
\\
s.t. ~~& \gamma \geq \alpha \geq \beta \\
& \cos(A) \leq 1/2 \\
& (\alpha,\beta,\gamma) \in \Delta.
\end{align*}
where $f_1=x_1 - (\gamma/2-\beta/2)$ (when $i=1$) and 
$f_2=(\gamma/2-\beta/2)-x_2$ (when $i=2$). Both of the non-linear programs can be solved analytically using software for symbolic calculations. Without loss of generality we assume that $\alpha=1$, so that both non linear programs attain the maximum value $0$ and the maximizers are 
$\alpha= 1, \beta = \frac{1}{128} \left(69-\sqrt{2101}\right), \gamma=\frac{69}{64}$ , and 
$\alpha = 1, \beta = \frac{1}{2}, \gamma= \frac{1}{4} \left(\sqrt{13}+1\right)$,
for $i=1,2$, respectively.
We conclude that $x_1 \leq \gamma/2-\beta/2 \leq x_2$.

It follows that, when $A\geq \pi/3$, we have that $T(x)=T_1(x)$ for all $x\in [x_1, x_2]$. Next we show that the minimizer $y$ of $T_1(x)$ lies in the interval $[x_1, x_2]$. But then, that would mean that $y$ is a local minimizer for $T(x)$ too. Since also by Lemma~\ref{cor: convexity of T} $T(x)$ is convex, it follows that $\min_{0\leq x\leq \alpha/2}T(x) = T(y)$. So it remains to find $y$, and show that $y \in [x_1,x_2]$ as promised. 

To that end, we compute the critical values of $T_1(x)$ in $[0,\alpha/2]$, by solving $T_1'(x)=0$. The solution to the latter equation are 
$$
y_{1,2}=\frac{\gamma}2-\frac{\beta}2\left(
\cos(A)\pm \frac{\sqrt3}3\sin(A)
\right).
$$
It is easy to see that when $A\geq \pi/3$, we have that $\cos(A) - \frac{\sqrt3}3\sin(A)\leq 0$, and that means that $y_2\geq \gamma/2$. Since also the domain of $T(x)$ is $[0,\alpha/2]$ and $\gamma\geq \alpha$, we conclude that the minimizer of $T_1(x)$ is attained at $x=y_1$. Finally, we calculate
\begin{align*}
T_1(y_1) 
& = \frac{\alpha}2+\frac{\gamma}2+\frac{\beta}2
\left(
\sqrt3 \sin(A) - \cos(A)
\right) \\
& =\frac{\sqrt{3} \sqrt{-((\alpha -\beta -\gamma ) (\alpha +\beta -\gamma ) (\alpha -\beta +\gamma ) (\alpha +\beta +\gamma ))}+(\alpha +\gamma )^2-\beta ^2}{4 \gamma }
\end{align*}
Recalling Heron's formula $\sqrt{p(p-\alpha)(p-\beta)(p-\gamma)}$ for the area $\tau$ of a triangle with half perimeter $p$, gives the promised lower bound 
when $A\geq \pi/3$. 

To conclude, we showed that 
$$
\min_{0\leq x\leq \alpha/2} T(x)  =
\left\{
\begin{array}{ll}
\alpha/2+\beta/2+\gamma/2 &~,\textrm{ if}~A\leq \pi/3 \\
\frac{(\alpha+\gamma)^2-\beta^2+4\sqrt{3} \tau}{4\gamma} &~,\textrm{ if}~A\geq \pi/3 
\end{array}
\right.
$$
The claim of the lemma follows by verifying that, as long as $\gamma\geq \alpha\geq \beta$ we have that 
$$
\frac{\alpha+\beta+\gamma}2 \geq \frac{(\alpha+\gamma)^2-\beta^2+4\sqrt{3} \tau}{4\gamma}
$$
if and only if $A\leq \pi/3$, i.e. if and only if $\cos{A} \geq 1/2$. 
\qed \end{proof}

\ignore{
{lem: some values of R1R2}
{lem: convexity of R1R2}
{cor: convexity of T}
}

\subsection{The case $\gamma \geq \beta \geq \alpha$}
\label{sec: cba}

As in the previous section, the domain of our function in Configuration 1 is either empty, or it is covered by Configuration 2 (when $\alpha=\gamma$). 

Next we move to Configuration 2. Recall that $x\leq \alpha/2$, and since this configuration happens only when $x\leq \gamma/2$ and $\gamma\geq \alpha$ we need to determine the extreme values of 
$$
T(x)=\alpha/2+x +\max\{CR_1, \beta,\gamma-2x\},
$$
in the interval $x\in [0,\alpha/2]$.

\begin{lemma}
\label{lem: cba conf2}
Given that $\gamma \geq \beta \geq \alpha$, we have that 
\begin{align*}
\max_{0\leq x\leq \alpha/2} T(x) & = \alpha/2 +\max\{
\gamma,\alpha/2+\beta
\}  ~~(\textrm{attained at}~x=0,\alpha/2,~\textrm{resp.})\\
\min_{0\leq x\leq \alpha/2} T(x) & =\alpha/2+\beta/2+\gamma/2 ~~(\textrm{attained at}~x=(\gamma-\beta)/2),
\end{align*}
\end{lemma}

\begin{proof}
The argument for the upper bound is given in the proof of Lemma~\ref{lem: cab conf2} (which holds as long as $\gamma\geq \max\{\alpha,\beta\}$. Next we provide the lower bound relying specifically on that $\gamma\geq \beta \geq \alpha$. Note that the latter inequality shows that $A$ is the smallest angle in $ABC$, and hence $A\leq \pi/3$. But then, 
Lemma~\ref{lem: T1 is less T2,T2} applies according to which 
$T(x)=\alpha/2+x+\max\{\beta,\gamma-2x\}$. 
Then we observe that $\beta=\gamma-2x$ at $x_0=\gamma/2-\beta/2\leq \alpha/2$ (due to the triangle inequality), and hence $T(x)$ is minimized at $x_0$. It is also easy to see that $T(x_0) = \alpha/2+\beta/2+\gamma/2$. 
\qed \end{proof}

\section{Evacuation Cost Bounds Starting from Anywhere on the Perimeter}
\label{sec: bounds from anywhere on the perimeter}

For evacuation algorithm \os, and 
for each $i\in \{a,b,c\}$, we denote by
$l_i(a,b,c)$ and $u_i(a,b,c)$
the smallest possible and largest possible worst-case evacuation cost, respectively, when the evacuation algorithm is restricted to start agents anywhere on edge $i$ (where the worst-case cost is over all placements of the exit). From the definition of the four algorithmic problems we introduced we have that 
\begin{align*}
\ll & \leq  \min_{i \in \{a,b,c\} } l_i (a,b,c) ~~\textit{(Overall smallest worst-case cost)}\\
\lu & \leq \max_{i \in \{a,b,c\} } l_i (a,b,c) ~~\textit{(Worst starting-edge for best starting-point worst-case cost)}\\
\ul & \leq \min_{i \in \{a,b,c\} } u_i (a,b,c)  ~~\textit{(Best starting-edge for worst starting-point worst-case cost)},\\
\uu & \leq \max_{i \in \{a,b,c\} } u_i (a,b,c) ~~\textit{(Overall largest worst-case cost)}
 \end{align*}
\ignore{
For example, $\ul$ is the best possible performance (as a function of the triangle edges) for the algorithmic problem of choosing an edge on the given triangle, and then letting the adversary choose both the starting point on the edge and the (worst-case) placement of the exit. 
In contrast, $\lu$ is the best possible performance 
for the algorithmic problem in which an adversary chooses the starting edge, and then allows the algorithm to pick the starting point on that edge so as to minimize the worst-case evacuation cost, over all placements of the exit.  
}

In this section we are concerned with the exact evaluation of $l_i(a,b,c)$ and $u_i(a,b,c)$ of algorithm \os, i.e. the smallest possible and highest possible worst-case evacuation cost, respectively, when starting anywhere on edge $i\in \{a,b,c\}$ of a non-obtuse triangle $ABC$, where $a\geq b \geq c$. 
Recall that the analysis we performed in Section~\ref{sec: analysis for special initial points} was for the smallest possible and highest possible worst-case evacuation cost on a triangle with edges $\alpha,\beta,\gamma$ (of arbitrary relevant lengths). Moreover, the setup configuration was that the agents' starting point was always on the half segment of edge $\alpha$, that was closer to the intersection of $\gamma$ with $\alpha$, than to the intersection of $\beta$ with $\alpha$. 

\begin{lemma}
\label{lem: start from largest}
For the smallest and largest possible worst-case evacuation cost of \os, when starting from the largest edge $a$, we have
\begin{align*}
l_a(a,b,c) &= \min\left\{ b+c, \frac{(a+b)^2-c^2+4\sqrt3 \tau}{4b}  \right\}  \\
u_a(a,b,c) & = \frac{a}{2} + b + \frac{c}{2},
\end{align*}
where $\tau$ is the area of the given triangle.
\end{lemma}

\begin{proof}
In this case we start on edge $a$, which is the same as starting on edge $\alpha$ in our previous analysis. Hence, we set $\alpha=a$. 
The starting point on $a$ can be either closer to the intersection of $a$ with $c$ or close to the intersection of $a$ with $b$. 

When the starting point is closer to the intersection of $a$ with $c$ (i.e. closer to $B$), and since $a\geq b \geq c$, the best possible and worst possible worst-case costs are given by Lemma~\ref{lem: abc conf1} and Lemma~\ref{lem: abc conf2}, with the assignment of $\alpha=a, \beta=b, \gamma=c$.  
On the other hand, 
if the starting point is closer to the intersection of $a$ with $b$ (i.e. closer to $C$), then 
the best possible and worst possible worst-case costs are given by Lemma~\ref{lem: acb conf1} and Lemma~\ref{lem: acb conf2}, with the assignment of $\alpha=a, \beta=c, \gamma=b$.

Therefore, we have
$$
l_a(a,b,c)
=
\min
\left\{
b+c, \frac{a^2-b^2+ac+2bc}{2c}, c+b, \frac{(a+b)^2-c^2+4\sqrt3 \tau}{4b}
\right\}.
$$
We claim that $\frac{a^2-b^2+ac+2bc}{2c} \geq \frac{(a+b)^2-c^2+4\sqrt3 \tau}{4b}$, for all non-obtuse triangles. To see why, we note that 
\begin{align*}
& \frac{a^2-b^2+ac+2bc}{2c} - \frac{(a+b)^2-c^2+4\sqrt3 \tau}{4b} \\
= & \frac{-c \sqrt{6 a^2 \left(b^2+c^2\right)-3 a^4-3 \left(b^2-c^2\right)^2}+a^2 (2 b-c)+3 b^2 c-2 b^3+c^3}{4 b c}.
\end{align*}
It can be seen that $-2 b^3 + a^2 (2 b - c) + 3 b^2 c + c^3\geq 0$, so that it suffices to show that the next expression is non-negative
\begin{align*}
& \left( -2 b^3 + a^2 (2 b - c) + 3 b^2 c + c^3 \right)^2
-c^2 \left( 
6 a^2 \left(b^2+c^2\right)-3 a^4-3 \left(b^2-c^2\right)^2
\right) \\
= & (-a^2 + b^2 - b c + c^2) (b^4 - 2 b^3 c + b c^3 + c^4 - 
   a^2 (b^2 - b c + c^2)).
 \end{align*}
Our claim then follows by showing that $-a^2 + b^2 - b c + c^2\leq 0$ and $b^4 - 2 b^3 c + b c^3 + c^4 -    a^2 (b^2 - b c + c^2) \leq 0$, concluding that 
$$
l_a(a,b,c) = \min\left\{ b+c, \frac{(a+b)^2-c^2+4\sqrt3 \tau}{4b}  \right\} 
$$

Now for calculating $u_a(a,b,c)$, and as described before, we invoke Lemma~\ref{lem: abc conf1} and Lemma~\ref{lem: abc conf2} with $\alpha=a, \beta=b, \gamma=c$, and Lemma~\ref{lem: acb conf1} and Lemma~\ref{lem: acb conf2} with $\alpha=a, \beta=c, \gamma=b$. Taking the maximum of all expressions, we have
$$
u_a(a,b,c)
=
\max\left\{
\frac{a}2+b+\frac{c}2, \frac{3a}{2}, \frac{a}{2}+b+\frac{c}2, \frac{a}2+c+\frac{b}2,\frac{3a}2
\right\}
=\frac{a}{2}+ \max\left\{ a, b+\frac{c}{2} \right\},
$$
where the last simplification is due to that $a\geq b \geq c$.
Now taking into consideration that the given triangle is non-obtuse, one case show that $a\leq b+c/2$, concluding that 
$ u_a(a,b,c) = \frac{a}{2} + b + \frac{c}{2}$. 
\qed \end{proof}
It is worthwhile mentioning here that $l_a(a,b,c)$ cannot be further simplified. Indeed, for $a=\sqrt{2}, b=c=1$, the formula is given by $l_a(a,b,c)=b+c=2$. 
On the other hand, when $a=b=c=1$, we have that $l_a(a,b,c)=
\frac{(a+b)^2-c^2+4\sqrt3 \tau}{4b}=3/2$ (while formula $b+c$ would evaluate to 2).

\begin{lemma}
\label{lem: start from medium}
For the smallest and largest possible worst-case evacuation cost of \os, when starting from the second largest edge $b$, we have
\begin{align*}
l_b(a,b,c) &= \max\left\{ 
\frac{a+b+c}{2},
\frac{(a+b)^2 - c^2 +4\sqrt3\tau}{4a}
 \right\}  \\
u_b(a,b,c) & = 
a+\frac{b}2 + \frac{c}2,
\end{align*}
where $\tau$ is the area of the given triangle. 
\end{lemma}

\begin{proof}
In this case we start on edge $b$, either close to point $A$ (intersection of edges $b,c$) or to point $C$ (intersection of edges $a,b$). 
When starting closer to $A$, the incident edge is $c$, which in our case is the smallest edge. Therefore, the best possible and worst possible worst-case costs in this case are given by 
Lemma~\ref{lem: bac conf1}
and Lemma~\ref{lem: bac conf2}
with $a=\beta, b=\alpha$ and $c=\gamma$, i.e. for our analysis in which $\beta\geq \alpha \geq \gamma$. 
On the other hand, when we start close to vertex $C$, the incident edge is $a$, which is the largest edge in the triangle. Therefore, the best possible and worst possible worst-case costs in this case are given by 
Lemma~\ref{lem: cab conf2}
with $a=\gamma, b=\alpha$ and $c=\beta$, i.e. for our analysis in which $\gamma\geq \alpha \geq \beta$. 

Towards computing $l_b(a,b,c)$ we need to evaluate the minimums of all lower bounds given by the previous theorems (under the underlying interpretation of $\alpha,\beta,\gamma$). A direct application of the lemmata gives
\begin{equation}
\label{equa: to simplify lb(a,b,c)}
l_b(a,b,c)
=
\min\left\{
a+c, 
\frac{b}2+a,
\max\left\{ 
\frac{a+b+c}{2},
\frac{(a+b)^2 - c^2 +4\sqrt3\tau}{4a}
 \right\} 
\right\}
\end{equation}
We proceed with a simplification of the previous expression. 

As per the arguments in the proof of Lemma~\ref{lem: cab conf2}, we have that 
$$
\max\left\{ 
\frac{a+b+c}{2},
\frac{(a+b)^2 - c^2 +4\sqrt3\tau}{4a}
 \right\} 
 =
\left\{
\begin{array}{ll}
\frac{a+b+c}{2} &~,\textrm{ if}~B\leq \pi/3 \\
\frac{(a+b)^2 - c^2 +4\sqrt3\tau}{4a} &~,\textrm{ if}~B\geq \pi/3 
\end{array}
\right.
$$
Therefore towards analyzing~\eqref{equa: to simplify lb(a,b,c)}, we examine the cases $B\leq \pi/3$ and $B> \pi/3$.
When $B\leq \pi/3$, we have that 
$$
l_b(a,b,c)
=
\min\left\{
a+c, 
\frac{b}2+a,
\frac{a+b+c}{2}
\right\}=\frac{a+b+c}{2},
$$
where the last simplification is due to that $a\geq b \geq c$. 
When $B> \pi/3$, we have that 
$$
l_b(a,b,c)
=
\min\left\{
a+c, 
\frac{b}2+a,
\frac{(a+b)^2 - c^2 +4\sqrt3\tau}{4a}
\right\},
$$
Recall that in this case we have $\cos(B) \leq 1/2$ and that $a\geq b \geq c$ and 
$(a,b,c) \in \Delta$ (for the edges to form a triangle). Under these constraints, one can form non-linear programs and verify that 
$$
a+c \geq \frac{(a+b)^2 - c^2 +4\sqrt3\tau}{4a},~~\textrm{and}~~
\frac{b}2+a \geq 
\frac{(a+b)^2 - c^2 +4\sqrt3\tau}{4a},
$$
as promised (where the inequalities become tight for $a=b=1, c=0$ and $a=b=c=1$, respectively).

Towards computing $l_b(a,b,c)$ we need to evaluate the minimums of all lower bounds given by the previous theorems (under the underlying interpretation of $\alpha,\beta,\gamma$). A direct application of the lemmata gives
$$
u_b(a,b,c)
=
\max\left\{
\frac{b}2+a+\frac{c}2, 
\frac{b}2+a,
b+c
\right\}
=
a+\frac{b}2+\frac{c}2,
$$
where the last simplification is due to that $a\geq b \geq c$. 
\qed \end{proof}
We note that the formula for $l_b(a,b,c)$ simplifies to a simpler formula (one of the arguments of $\max\{\cdot\}$) depending on whether $B\geq \pi/3$ or not).

\begin{lemma}
\label{lem: start from smallest}
For the smallest and largest possible worst-case evacuation cost of \os, when starting from the smallest edge $c$, we have
\begin{align*}
l_c(a,b,c) &= 
\frac{a+b+c}2  \\
u_c(a,b,c) & = 
a+ c.
\end{align*}
\end{lemma}

\begin{proof}
The starting edge here is $c$, which plays the role of edge $\alpha$ in our original analysis of Section~\ref{sec: analysis for special initial points}. 
There are two cases to consider. Either the starting point is close to $A$ (intersection of $c$ and $b$), or it is closer to $B$ (intersection of $a$ and $c$). 

When starting from the a point on $c$ closer to $A$, we start closer to the incident edge $b$, which is the second largest edge in the given triangle. Since in the analysis of Section~\ref{sec: analysis for special initial points} we always start closer to the intersection point with edge $\gamma$, it follows that the best possible and worst possible worst-case costs in this case are given by 
Lemma~\ref{lem: bca conf2},
setting $a=\beta, b=\gamma, c=\alpha$ (i.e. the case $\beta\geq \gamma\geq \alpha$). 

When starting from the a point on $c$ closer to $B$, we start closer to the incident edge $a$, which is the largest edge in the given triangle. Since in the analysis of Section~\ref{sec: analysis for special initial points} we always start closer to the intersection point with edge $\gamma$, it follows that the best possible and worst possible worst-case costs in this case are given by 
Lemma~\ref{lem: cba conf2},
setting $a=\gamma, b=\beta, c=\alpha$ (i.e. the case $\gamma \geq \beta \geq \alpha$). 

Towards computing $l_c(a,b,c)$ we evaluate the minimum of all lower bounds given by the previous lemmata (under the underlying interpretation of $\alpha,\beta,\gamma$). A direct application of the lemmata gives
$$
l_c(a,b,c)
=
\min\left\{
\frac{c}2+a,
\frac{a+b+c}2
\right\}
= \frac{a+b+c}2,
$$
where the last simplification is due to that $a\geq b \geq c$. 

Similarly, towards computing $u_c(a,b,c)$ we evaluate the maximum of all upper bounds given by the previous lemmata (under the underlying interpretation of $\alpha,\beta,\gamma$). A direct application of the lemmata gives
$$
u_c(a,b,c)
=
\max
\left\{
c+a,
\frac{c}2+a,
c+b
\right\}
=a+ c,
$$
where the last simplification is due to that $a\geq b \geq c$. 
\qed \end{proof}

We are now ready to prove the upper bounds of $\ll, \lu, \ul$ and $\uu$ as stated in Theorem~\ref{thm: four problems}
\ignore{
\begin{theorem}
\label{thm: four problems}
For triangle $ABC$ with edges $a\geq b \geq c$, we have that 
\begin{align*}
\ll & = \frac{a+b+c}2,\\
\lu & = \min\left\{ b+c, \frac{(a+b)^2-c^2+4\sqrt3 \tau}{4b}  \right\},  \\
\ul & = \min\left\{
\frac{a}{2} + b + \frac{c}{2}, 
a+ c
\right\}
,\\
\uu & = a+ \frac{b+c}2,
 \end{align*}
where $\tau$ is the area of the triangle. 
\end{theorem}

We are now ready to prove Theorem~\ref{thm: four problems}.
}
\begin{proof}[Proof of the Upper bounds of Theorem~\ref{thm: four problems}]
By Lemmata~\ref{lem: start from largest},~\ref{lem: start from medium} and~\ref{lem: start from smallest}, we know all $l_i(a,b,c), u_i(a,b,c)$ for $i\in       \{a,b,c\}$. Therefore we compute
\begin{align*}
\uu  & \leq \max_{i \in \{a,b,c\} } u_i (a,b,c) \\
& = \max\left\{
\frac{a}{2} + b + \frac{c}{2}, 
a+\frac{b+c}2,
a+ c
\right\} \\
& = a+ a+ \frac{b+c}2.
\end{align*}
We also have 
\begin{align*}
\ul  & \leq \min_{i \in \{a,b,c\} } u_i (a,b,c) \\
& = \min\left\{
\frac{a}{2} + b + \frac{c}{2}, 
a+\frac{b+c}2,
a+ c
\right\} \\
& = 
\min\left\{
\frac{a}{2} + b + \frac{c}{2}, 
a+ c
\right\}
\end{align*}

Next we calculate $\min_{i \in \{a,b,c\} } l_i (a,b,c)$, which is an upper bound to $\ll$. By the already established results, we have that 
$$
\ll
\leq
\min
\left\{
\begin{array}{l}
b+c \\
\frac{(a+b)^2-c^2+4\sqrt3 \tau}{4b} \\
\frac{a+b+c}{2}\\
\max\left\{ 
\frac{a+b+c}{2},
\frac{(a+b)^2 - c^2 +4\sqrt3\tau}{4a}
 \right\} 
\end{array}
\right.
=
\min
\left\{
\begin{array}{l}
b+c \\
\frac{(a+b)^2-c^2+4\sqrt3 \tau}{4b} \\
\frac{a+b+c}{2}\\
\end{array}
\right.
=
\min
\left\{
\begin{array}{l}
\frac{(a+b)^2-c^2+4\sqrt3 \tau}{4b} \\
\frac{a+b+c}{2}\\
\end{array}
\right.,
$$
where the last equality is due to that $b+c\geq (a+b+c)/2$, by the triangle inequality. 

Next we show that 
$
\frac{(a+b)^2-c^2+4\sqrt3 \tau}{4b} \geq \frac{a+b+c}{2}
$, which will conclude our claim about $\ll$ (but we will use this inequality in our last argument too). To see why we calculate
\begin{align*}
& \frac{(a+b)^2-c^2+4\sqrt3 \tau}{4b} - \frac{a+b+c}{2} \\
& = \frac{a^2+\sqrt{3} \sqrt{-((a-b-c) (a+b-c) (a-b+c) (a+b+c))}-(b+c)^2}{4 b}.
\end{align*}
Clearly it is enough to show that the next expression is non-negative
$$
3 (b+c-a) (a+b-c) (a-b+c) (a+b+c) - \left(  (a + b)^2 - c^2  \right)^2
=
4 (b+c-a) (a+b+c) \left(a^2-b^2+b c-c^2\right). 
$$
The last expression is indeed non-negative, since $b+c-a \geq 0$ by the triangle inequality, and since by the Cosine Law $a^2-b^2+b c-c^2 = bc (1-2\cos(A) )\geq 0$, where the last inequality is because $a$ is the largest edge, hence $A\geq \pi/3$, hence $\cos(A)\leq 1/2$.

Finally, we calculate $\max_{i \in \{a,b,c\} } l_i (a,b,c)$, which is an upper bound to $\lu$. For this, we have
$$
\lu
\leq
\max
\left\{
\begin{array}{l}
\min
\left\{b+c, \frac{(a+b)^2-c^2+4\sqrt3 \tau}{4b} \right\} \\
\frac{a+b+c}{2}\\
\max\left\{ 
\frac{a+b+c}{2},
\frac{(a+b)^2 - c^2 +4\sqrt3\tau}{4a}
 \right\} 
\end{array}
\right.
=
\max
\left\{
\begin{array}{l}
\min
\left\{b+c, \frac{(a+b)^2-c^2+4\sqrt3 \tau}{4b} \right\} \\
\frac{a+b+c}{2}\\
\frac{(a+b)^2 - c^2 +4\sqrt3\tau}{4a}
\end{array}
\right.
$$

From the discussion above we know that $b+c \geq (a+b+c)/2$ and 
$$
\frac{a+b+c}{2} \leq \frac{(a+b)^2-c^2+4\sqrt3 \tau}{4b}, 
$$
and hence $(a+b+c)/2 \leq \min \left\{b+c, \frac{(a+b)^2-c^2+4\sqrt3 \tau}{4b}\right\}$.
So, we have 
$$
\lu
\leq
\max
\left\{
\begin{array}{l}
\min \left\{b+c, \frac{(a+b)^2-c^2+4\sqrt3 \tau}{4b}\right\}
\\
\frac{(a+b)^2 - c^2 +4\sqrt3\tau}{4a}
\end{array}
\right..
$$
Lastly, we recall that $a\geq b$, and hence $\frac{(a+b)^2 - c^2 +4\sqrt3\tau}{4b}\geq \frac{(a+b)^2 - c^2 +4\sqrt3\tau}{4a}$. Our main claim pertaining to $\lu$ will follow once we show that $b+c\geq \frac{(a+b)^2 - c^2 +4\sqrt3\tau}{4a}$. To that end we we consider the non-linear program 
\begin{align*}
\min&~~ (b+c) - \frac{(a+b)^2 - c^2 +4\sqrt3\tau}{4a}
\\
s.t. ~~& a=1 \\
& (a,b,c) \in \Delta.
\end{align*}
whose optimal value is 0, attained among others at $a=b=1, c=0$. 
\qed \end{proof}

\section{Lower Bounds}
\label{sec: lower bounds}
This section is concerned with proving the 
lower bounds of $\ll, \lu, \ul$ and $\uu$ as stated in Theorem~\ref{thm: four problems}.

\ignore{
\begin{theorem}
For a triangle with edges $a\geq b \geq c$, the following bounds hold.  \\
\label{thm: lower bound}
$
\ll \geq (a+b+c)/2
$. \\
$
\uu \geq a+(b+c)/2
$
\end{theorem}
}

\begin{proof}[Proof of the Lower bounds of Theorem~\ref{thm: four problems}]
In all our arguments below, we denote by $A,B,C$ the vertices opposite to edges $a\geq b \geq c$, respectively.

\paragraph{The lower bound of $\ll$:}
The proof follows by observing that $(a+b+c)/2$ is half the perimeter of the geometric object where the exit is hidden. Indeed, consider any algorithm for the triangle evacuation problem, where agents can possibly choose even distinct starting location on the perimeter of the given triangle. In time $(a+b+c)/2-\epsilon/2$, the two agents can search at most $a+b+c- \epsilon$, and hence there will also be, at that moment, an unexplored point. In other words, for every $\epsilon >0$, the evacuation time is at least $(a+b+c)/2- \epsilon$.

\paragraph{The lower bound of $\lu$:}
In this problem, the adversary chooses a starting edge, and then the algorithm may deploy the two robots on any point on that edge before attempting to evacuate the given triangle (from an exit that is again chosen by the adversary). An adversary who also allows the algorithm to choose the starting edge is even weaker, and therefore the lower bound of $\ll$ is a lower bound to $\lu$ too.

\paragraph{The lower bound of $\ul$:}
Recall that in the underlying algorithmic problem, the algorithm chooses the starting edge, while the adversary chooses the starting point on the edge. So we consider some cases as to which starting edge an algorithm may choose. If the starting edge is either $b$ or $c$, then the adversary can choose starting point $A$. Consequently, the adversary can wait until at the first among $B,C$ is visited, and place the exit in the other vertex (the argument works even if vertices are visited simultaneously). But then, the evacuation time would be 
$$
\min\{AB,AC\} + BC = \min\{b,c\}+a=c+a.
$$
In the last case, the algorithm may choose to start from edge $a$. In that case the adversary may choose strarting point $D$ which is $(a-c)/2$ away from $B$, i.e. for which $BD=(a-c)/2$ and $DC=(a+c)/2$.
For that starting point, an adversary can wait until the first vertex among $A,C$ is visited. If $C$ is the first (which happens at time at least $DC=(a+c)/2$), the exit can be placed at $A$, then the induced running time would be at least $b+(a+c)/2$. Overall that would mean that the performance of the arbitrary algorithm would be $\min\{c+a, b+(a+c)/2\}$, which also equals our upper bound. 
If, on the other hand, $A$, is visited before $C$, then that can happen in time at least 
$$
AD= \left(
c^2 + (a-c)^2/4-c(a-c) \coss{B}
\right)^{1/2},
$$
where also $\cos(B) = \frac{-b^2+a^2+c^2}{2 a c}$ (note that $AD \leq (a+c)/2$ by the triangle inequality). 
Then, placing the exit at $C$ induces evacuation time at least $AD+b$. 
Overall, since the algorithm can choose the starting edge, the induced lower bound to $\ul$ is at least 
$$
\min\{
a+c, AD+b
\}
=
\min\left\{
a+c, \frac{1}{2} \sqrt{\frac{2 b^2 (a-c)-(a-2 c) (a+c)^2}{a}} +b
\right\}.
$$
If we call the latter expression $f(a,b,c)$, then a Non-Linear Program can establish that 
$$f(a,b,c)\geq \frac{1}{10} \left(\sqrt{10}+5\right) \min\{c+a, b+(a+c)/2\}$$ (over all triangles with $A\leq \pi/2$), meaning that our lower bound is off by at most $\left(\sqrt{10}+5\right)\approx 0.816228$ from our upper bound. 
By letting also denote by $h(t)$ the smallest ratio $f(a,b,c)/\min{c+a, b+(a+c)/2}$ (over all triangles non-obtuse triangles) condition on that $c/a=t\in (0,1]$, we can also derive that our lower bound is much better than $0.816228$ times our upper bound, for the various values of $t$, see Figure~\ref{fig: fOverupperboundUL}
\begin{figure}[h!]
  \centering
  \includegraphics[width=0.5\linewidth]{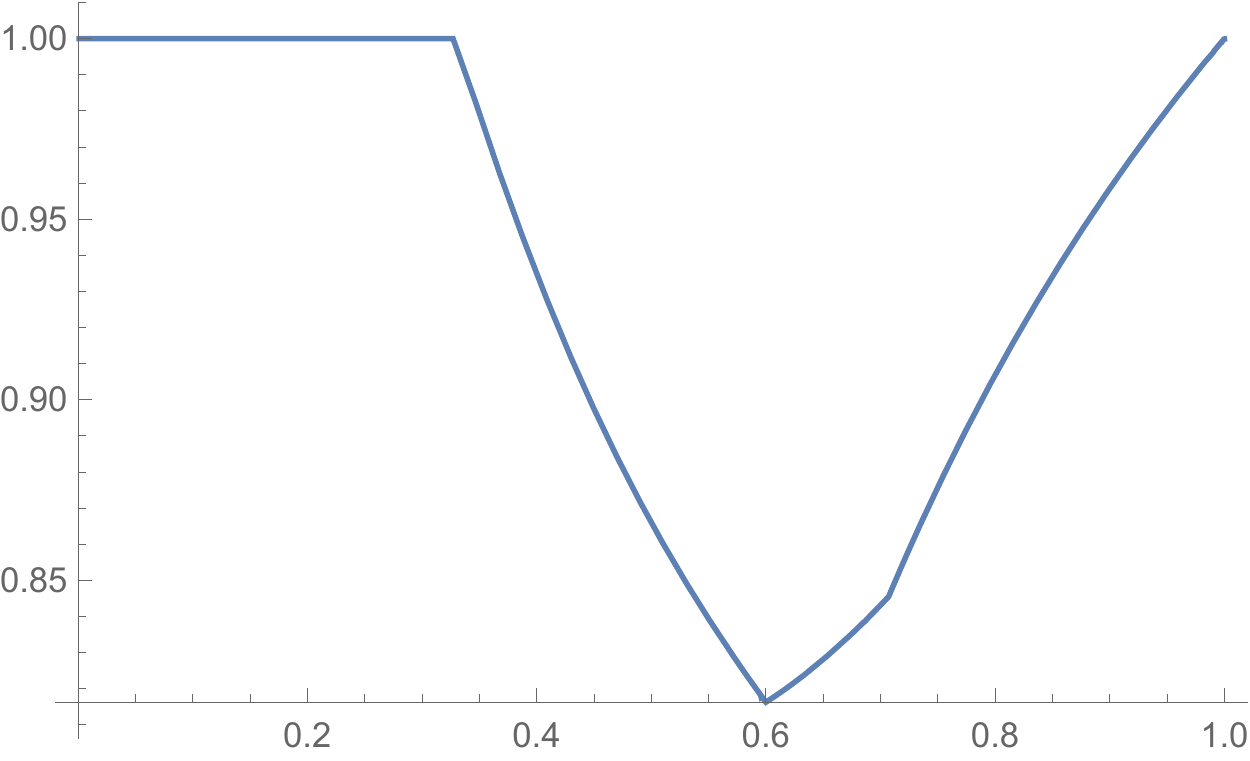}
\caption{The plot of $h(t) = \min_{(a,b,c)\in \Delta} f(a,b,c)/\min\{c+a, b+(a+c)/2\}$ that quantifies how far away is our lower bound from the derived upper bound to $\ul$, as a function of $c/a\in (0,1]$. }
\label{fig: fOverupperboundUL}
\end{figure}

\ignore{
hh[a_, b_, c_] := 1/2 Sqrt[(2 a^2 (b - c) - (b - 2 c) (b + c)^2)/b]

ddd = 1/2 Sqrt[(-a^3 + 2 a b^2 - 2 b^2 c + 3 a c^2 + 2 c^3)/a]
ddd=hh[b, a, c]

upper[a_, b_, c_] := Min[ a/2 + b + c/2, a + c];
lower[a_, b_, c_] := Min[ (ddd + b), a + c];

Minimize[ {lower[a, b, c]/upper[a, b, c] , 
  a >= 0 && b >= 0 && c >= 0 && a + b >= c && a + c >= b && 
   b + c >= a && a >= b >= c && coss[a, b, c] >= 0}, {a, b, c}]

test2[t_] := 
 Minimize[ {lower[a, b, c]/upper[a, b, c] , 
   a >= 0 && b >= 0 && c >= 0 && a + b >= c && a + c >= b && 
    b + c >= a && a >= b >= c  && c/a == t && a == 1 && 
    coss[a, b, c] >= 0}, {a, b, c}]

Plot[ test2[tt][[1]], {tt, 0.001, 1}]
 
}

\paragraph{The lower bound of $\uu$:}
For this problem, the adversary can choose the agents' starting point on the perimeter of the triangle $ABC$.
Consider the starting point $D$, on edge $b$, such that $AD=(b-c)/2$, and note that by the Cosine Law in triangle $ADB$ we have that 
$$
BD = \left(
c^2 + (b-c)^2/4-c(b-c) \coss{A}
\right)^{1/2}
$$
where also $\cos(A) = \frac{-a^2+b^2+c^2}{2 b c}$ (note that $BD \leq (b+c)/2$ by the triangle inequality). 
Now, with starting point $D$, we let an arbitrary algorithm run until the first point among $B,C$ is visited (our argument works even if the points are visited simultaneously). If $C$ is visited first, and since $DC=(b+c)/2$, this does not happen before time $(b+c)/2$, and so by placing the exit at $B$, the evacuation time would be at least $a+(b+c)/2$ (and that would be equal to the established upper bound). 
If on the other hand $B$ is discovered before $C$, then we can place the exit $C$ (when $B$ is visited by any robot), and that would induce running time at least 
$$
BD+a
=\frac{1}{2} \sqrt{\frac{2 a^2 (b-c)-(b-2 c) (b+c)^2}{b}}+a.
$$
If we call the last expression $f(a,b,c)$, then a Non-Linear Program can establish that 
$f(a,b,c) \geq\footnote{0.852 is the third smallest (also second largest) real root of polynomial $20 x^6 - 68 x^5 + 671 x^4 - 2776 x^3 + 2550 x^2 - 516 x - 25 $ (the polynomial has 4 real and 2 complex roots).} 0.852 \left( a + \frac{b+c}2\right)$. By denoting the smallest ratio $f(a,b,c)/\left( a + \frac{b+c}2\right)$ (over all triangles) condition on that $c/b=t\in (0,1]$ as $h(t)$, we can also derive that our lower bound is much better than $0.852$ times our upper bound, for the various values of $t$, see Figure~\ref{fig: fOverabsover2}.
\begin{figure}[h!]
  \centering
  \includegraphics[width=0.5\linewidth]{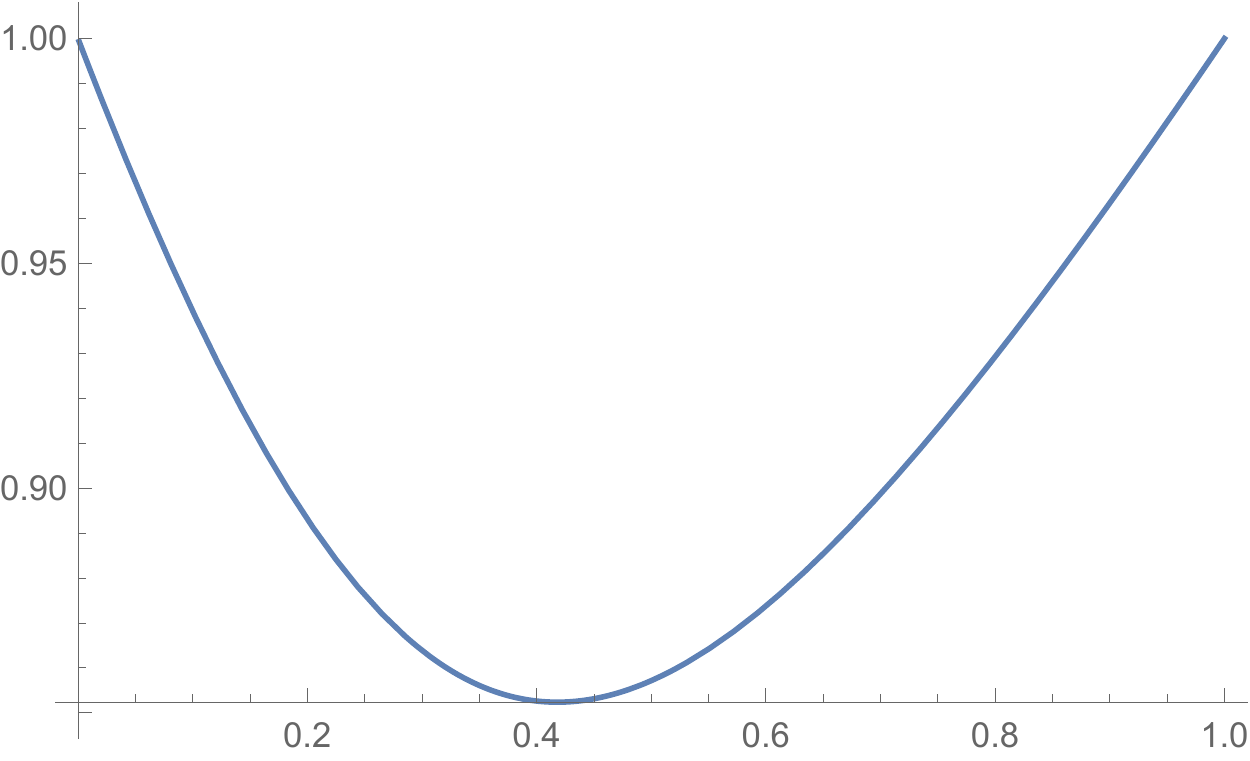}
\caption{The plot of $h(t) = \min_{(a,b,c)\in \Delta} f(a,b,c)/\left( a + \frac{b+c}2\right)$ that quantifies how far away is our lower bound from the derived upper bound to $\uu$, as a function of $c/b\in (0,1]$. }
\label{fig: fOverabsover2}
\end{figure}

\ignore{
dd = 1/2 Sqrt[(2 a^2 (b - c) - (b - 2 c) (b + c)^2)/b];

Minimize[ {(dd + a)/( (b + c)/2 + a) , 
  a >= 0 && b >= 0 && c >= 0 && a + b >= c && a + c >= b && 
   b + c >= a && a >= b >= c }, {a, b, c}]

test[t_] := 
 Minimize[ {(dd + a)/( (b + c)/2 + a) , 
   a >= 0 && b >= 0 && c >= 0 && a + b >= c && a + c >= b && 
    b + c >= a && a >= b >= c  && c/b == t && a == 1}, {a, b, c}]
    
Plot[ test[tt][[1]], {tt, 0.001, 1}]    
}

\qed \end{proof}

\section{Discussion}

We provided upper and lower bounds to the 2 searcher evacuation problem from triangles in the wireless model, extending results previously known only for equilateral triangles. Our main contribution is a technical analysis of a plain-vanilla algorithm that has been proven to be optimal for various search domains in the wireless model. We also provided lower bounds that depending on the given triangle could range from being tight to having a gap of at most $.8$ depending on the considered algorithmic problem and the given triangle. 
Providing tight upper and lower bounds for the entire spectrum of triangles is an open problem.
We also see the study of the same problem in the face-to-face model as the next natural direction to consider, but definitely more challenging.

\bibliographystyle{plain}

\bibliography{TriangleEvacuationBiblio}

\appendix

\end{document}